\documentclass[a4paper,USenglish,cleveref,autoref,thm-restate]{lipics-v2021}

\usepackage{array}
\usepackage{booktabs}
\usepackage{enumitem}      %
\usepackage{graphicx}
\usepackage{xspace}
\usepackage{tikz}
\usepackage[most]{tcolorbox}
\usepackage{adjustbox}
\usetikzlibrary{arrows.meta,calc,positioning,fit,backgrounds,decorations.pathreplacing,shadows}
\ifdefined\tikzexternalrealjob\def\externalizefigures{}\fi
\ifdefined\externalizefigures
  \usetikzlibrary{external}
\fi
\providecommand{\tikzsetnextfilename}[1]{}
\extrafloats{64}

\newif\ifexamples
\examplesfalse
\ifdefined\noexamples\examplesfalse\fi
\usepackage{comment}
\ifexamples
\newtcolorbox{runex}{enhanced, breakable, title={Running Example},
  fonttitle=\sffamily\bfseries\small, coltitle=black,
  colback=lipicsLightGray!20, colbacktitle=figYellow!55, colframe=figline,
  boxrule=0.4pt, arc=\arrrad, left=4pt, right=4pt, top=2pt, bottom=2pt,
  before skip=4pt, after skip=4pt,
  before upper={\setlength{\parindent}{0pt}\setlength{\parskip}{0.5\baselineskip}}}
\else
\excludecomment{runex}
\fi

\DeclareRobustCommand{\runningexample}{\ifexamples the running example\else
  $T = \texttt{GTGCTGAGAGATGCTT}$ with $w = 2$ and
  $\mathcal{E} = \{\texttt{CT}, \texttt{GA}\}$\fi}

\newcommand{\code}[1]{\texttt{#1}}
\newcommand{\Docc}{\mathcal{C}}      %
\newcommand{\Dset}{\mathcal{D}}      %
\newcommand{\Parse}{\mathcal{P}}     %
\newcommand{\Suf}{\mathcal{S}}       %
\newcommand{\Dict}{\mathcal{T}_{\Dset}}  %
\newcommand{\lex}{\prec}
\newcommand{\rev}[1]{\overleftarrow{#1}}
\newcommand{\suf}{\mathrm{suf}}
\newcommand{\pre}{\mathrm{pre}}
\newcommand{\rank}{\mathrm{rank}}
\newcommand{\SA}{\mathrm{SA}}
\newcommand{\BWT}{\mathrm{BWT}}
\newcommand{\bmax}{b_{\max}}
\newcommand{\lb}{\mathit{lb}}        %
\newcommand{\rb}{\mathit{rb}}        %
\newcommand{\libcubwt}{\code{libcubwt}\xspace}
\newcommand{\libsais}{\code{libsais}\xspace}
\newcommand{\gpupfpbwt}{\code{gpupfpbwt}\xspace}

\title{Engineering Lightweight Burrows-Wheeler Transform Construction on the GPU via Prefix-Free Parsing}
\titlerunning{Engineering Lightweight BWT Construction on the GPU via Prefix-Free Parsing}

\author{Florian Kurpicz}{University of Münster, Germany}{kurpicz@uni-muenster.de}{https://orcid.org/0000-0002-2379-9455}{}
\author{Maximilian Paris}{University of Münster, Germany}{mparis1@uni-muenster.de}{}{}
\authorrunning{F.\ Kurpicz and M.\ Paris}

\Copyright{Florian Kurpicz and Maximilian Paris}
\ccsdesc[500]{Theory of computation~Data compression}
\ccsdesc[300]{Theory of computation~Sorting and searching}
\keywords{Burrows--Wheeler transform, prefix-free parsing, GPU, suffix sorting, memory-efficient algorithms}
\category{}
\relatedversion{Based on M.~Paris' Master's thesis~\cite{MaximilianParis2026}.}
\supplement{\url{https://codeberg.org/kurpicz/gpupfpbwt} (source code)}
\funding{}
\acknowledgements{}

\nolinenumbers
\hideLIPIcs

\EventEditors{TODO}
\EventNoEds{1}
\EventLongTitle{TODO}
\EventShortTitle{SEA 2027}
\EventAcronym{SEA}
\EventYear{2027}
\EventDate{TODO}
\EventLocation{TODO}
\EventLogo{}
\SeriesVolume{TODO}
\ArticleNo{TODO}

\begin{document}
\maketitle

\begin{abstract}
The Burrows--Wheeler transform (BWT) is the core of many compressed full-text indexes and most read aligners.
The fastest BWT constructions on GPUs need about $20$ bytes of device memory per character, which severely limits the texts that can be processed there.
For example, on our GPU with $24$\,GB device memory, this is only about $1\,200$\,MiB.
Prefix-free parsing (PFP) can be used to reduce the memory requirements of BWT construction algorithms in main memory.
But, it uses hash tables and inverted lists which are not suited for GPUs.
We show that PFP-based BWT construction can work with only sorting, scanning, or simple data-parallel procedures.
This results in our GPU construction \gpupfpbwt that on an NVIDIA RTX 3090, reduces the peak memory on an English text to $11.9$ bytes per character.
There, it can processes $2\,108$\,MiB input files, which is $1.8$ times larger than what the current state of the art \libcubwt can handle.
It is faster than or on par with \libcubwt on eight of nine repetitive texts but $2.0$ to $3.2$ times slower on five of six non-repetitive ones, i.e., it provides a trade-off between required device memory and throughput.
\end{abstract}

\section{Introduction}\label{sec:intro}

The Burrows--Wheeler transform (BWT)~\cite{BurrowsW94} of a text is the sequence of the characters that precede the text's suffixes in lexicographic order.
It is reversible and used in compression~\cite{BurrowsW94,Manzini01}.
A BWT consists of $r$ maximal runs of equal characters, which can be stored in $O(r)$ words.
Here, each run is represented by its character and length.
This makes $r$ a common measure of the compressibility of a text.
Since the BWT supports pattern matching in compressed form, it is the core of compressed full-text indexes such as the FM-index~\cite{FerraginaM00}, the r-index~\cite{GagieNP20}, and the move data structure~\cite{NishimotoT21}.
These indexes are used read aligners such as Bowtie~\cite{LangmeadTPS09}, BWA~\cite{LiD09}, and SOAP2~\cite{LiYLLYKW09} (FM-index) and pangenome indexes such as MONI~\cite{RossiOLGB22} (r-index), Movi~\cite{ZakeriBAGL24}, and Movi~2~\cite{ZakeriBGL26} (move data structure).

The BWT can be construed by sorting all suffixes of the text, a data-parallel task that suits GPUs.
The fastest GPU BWT construction, \libcubwt~\cite{libcubwt}, needs about $20$ bytes of device memory per character.
On a consumer GPU with $24$\,GB of memory, this limits \libcubwt to texts of about $1\,200$\,MiB, far below the size of the string collections and pangenomes that BWT-based indexes are built for.
In main memory, PFP-BWT, the construction based on prefix-free parsing (PFP)~\cite{BoucherGKLMM19}, reduces the required working memory.

\subparagraph*{Our Contributions.}
We port prefix-free parsing to the GPU.
Our key observation is that some structures, e.g., the hash tables and inverted lists of the original unnecessary.
If all phrase \emph{occurrences} are kept instead of the distinct phrases, we can obtain all necessary information just by sorting.
Every step of our algorithm thus becomes a sort, a scan, or a simple data-parallel kernel.
After the preliminaries (\cref{sec:prelim}), related work (\cref{sec:related}), and the original PFP-BWT algorithm (\cref{sec:pfp}), we present
\begin{itemize}[nosep]
\item a formulation of PFP-BWT that consists of sorts and scans only (\cref{sec:algo}),
\item memory levels that need less memory than \libcubwt (\cref{sec:memory,sec:ref-level56}),
\item two sorting cores and a space-saving configuration (\cref{sec:ref-msd,sec:ref-space,sec:ref-doubling}), and
\item an extensive experimental evaluation (\cref{sec:eval}).
\end{itemize}

\section{Preliminaries}\label{sec:prelim}

Let $T[0..n-1]$ be a text of length $n$ over an alphabet $\Sigma$ of size $\sigma$.
$T[i..j]$ denotes the substring from position $i$ to position $j$, $T[i..]$ the suffix starting at position~$i$, and $\rev{x}$ the reverse of a string $x$.
We write $x \lex y$ if $x$ is lexicographically smaller than $y$ and consider a proper prefix as smaller.
The suffix array $\SA$ lists the starting positions of the lexicographically sorted suffixes.
Then, $\BWT[i] = T[(\SA[i] - 1) \bmod n]$, i.e., the $i$th character is the one that cyclically precedes the $i$th smallest suffix.
We use the notation of \libcubwt~\cite{libcubwt} and return the \emph{primary index} $q$ with $\SA[q] = 0$ instead of appending a sentinel to the text.

For a base $z$ and a modulus $p$, the \emph{Karp--Rabin fingerprint}~\cite{KarpR87} of a string $x$ of length $\ell$ is $\phi(x) = \bigl(\sum_{i=0}^{\ell-1} z^{\ell-1-i} \cdot x[i]\bigr) \bmod p$.
The fingerprints of all windows of a fixed length can be computed in one left-to-right scan with a rolling hash.

\subparagraph*{GPU Programming.}
A GPU executes a \emph{kernel} with many thousands of threads.
Up to $1\,024$ threads form a \emph{thread block}, which runs on one streaming multiprocessor and shares its small on-chip \emph{shared memory}.
All threads access the \emph{device memory}, whose bandwidth ($936$\,GB/s on our GPU) is reached only if neighboring threads access consecutive memory.
The device memory ($24$\,GB on our GPU) is small compared to the host memory. %
We the \emph{peak resident device memory}, i.e., the maximum number of bytes allocated on the GPU at the same time.
Our algorithms combine kernels with the data-parallel primitives of CUB~\cite{CUB} and Thrust~\cite{BellH12}, e.g., prefix sums (\emph{scans}), discarding, reductions, and radix sorts.%

\section{Related Work}\label{sec:related}

Since the BWT follows from the suffix array in one scan, every suffix array construction algorithm (SACA) also constructs the BWT.
We discuss BWT construction and suffix sorting on the GPU. \Cref{sec:more_related_work} covers other models of computation, and a survey~\cite{BingmannDFKOS22} covers further building blocks of full-text indexes that utilize the BWT.

\subparagraph*{BWT construction on the GPU.}
The first GPU BWT constructions transform the blocks of bzip2-like compressors~\cite{PatelZMDO12,DeshpandeN15}.
Zhao et al.~\cite{ZhaoYX12} build the BWT of a human genome by blockwise suffix sorting~\cite{Karkkainen07}, with the GPU sorting the suffixes of each block.
Liu et al.~\cite{LiuLL14} builds the BWT of large read collections on a GPU together with the CPU cores and on clusters of GPU nodes.
The current state of the art, \libcubwt~\cite{libcubwt} uses suffix sorting routines with the radix sort of CUB (\cref{sec:eval}).
Closest to our work is Rapid-PFP~\cite{FerroPGLB26}, which computes the prefix-free parse on the GPU but not the BWT.

\subparagraph*{Suffix sorting on the GPU.}
GPU SACAs use prefix doubling on radix sorts~\cite{Osipov12}, the skew algorithm~\cite{DeoK13,LiaoMZT15}, hybrids of both~\cite{WangBO16}, CPU--GPU hybrids~\cite{Sun14}, lexicographic naming~\cite{SunM09}, genome-specific techniques~\cite{PingaliTB17}, and several GPUs connected by NVLink~\cite{BurenJKHS19}.
All of them sort all $n$ suffixes of the text.
We are not aware of one that uses induced sorting, the basis of the fastest sequential SACAs~\cite{SACABench19}.
In general, the suffix array requires more device memory than the BWT.
Beyond suffix sorting, other building blocks of compressed indexes have been engineered for GPUs, such as minimal perfect hash functions~\cite{BezKLS23} and wavelet trees~\cite{PiochowiakKD25,FranzrebBR25}.

\section{Prefix-Free Parsing and the Original PFP-BWT}
\label{sec:pfp}

Boucher et al.~\cite{BoucherGKLMM19} introduce a BWT construction algorithm based in prefix-free parsing (PFP).
IT replaces the suffix sort of the text by a sort of the distinct phrase suffixes, which are few for repetitive texts.
Our GPU algorithm builds on their construction.

\subsection{The Parse}
\label{sec:pfp-parse}
PFP is parameterized by a window length $w \ge 1$ and a set of \emph{trigger strings} $\mathcal{E}$ of strings of length $w$.
Two sentinel symbols $\#$ and $\$$, i.e., symbols smaller than every character of the alphabet, are prepended and appended.
The parse is then computed on $\#\,T\,\$^w$.

A \emph{phrase} is a substring of $\#\,T\,\$^w$ from the beginning of an occurrence of an element of $\mathcal{E}' = \mathcal{E} \cup \{\#, \$^w\}$ to the end of the next one.
Thus, a phrase has exactly one proper prefix and one proper suffix in $\mathcal{E}'$ and no other substring in $\mathcal{E}'$.
Consecutive phrases overlap by the $w$ characters.
The \emph{dictionary} $\Dset$ is the set of distinct phrases.
The \emph{parse} $\Parse$ is the sequence of the phrases in text order, each identified by its lexicographic rank in $\Dset$.

\subsection{The Structure of the Phrase Suffixes}\label{sec:pfp-suffixes}

The set $\Suf$ of \emph{phrase suffixes} consists of the suffixes of length $> w$ of the elements of $\Dset$.
Every phrase suffix ends with the trigger string that ends its phrase.
The following lemmas show why sorting $\Suf$ and the parse can replace the sort of all suffixes of $T$.

\begin{lemma}[{\cite[Lemma~1]{BoucherGKLMM19}}]\label{lem:bgk-1}
The set $\Suf$ is prefix-free.
\end{lemma}

\noindent If $\alpha \in \Suf$ were a proper prefix of $\alpha\gamma \in \Suf$, the trigger string that ends $\alpha$ would lie in the interior of a phrase $v$ that ends with $\alpha\gamma$, see \cref{fig:preffree}.
Our algorithm also relies on the phrases being suffix-free, which is implicit in the proof of~\cite[Lemma~7]{BoucherGKLMM19}.

\begin{lemma}\label{lem:sufffree}
Let $u, v \in \Dset$ with $|u| < |v|$.  Then $u$ is not a suffix of $v$.
\end{lemma}

\begin{figure}[tbp]
\centering
\begin{minipage}[t]{0.485\linewidth}
\centering
\adjustbox{max width=\linewidth}{\includegraphics{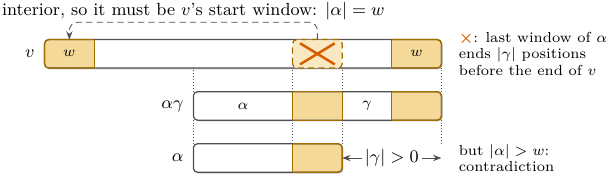}}
\caption{Visualization of \cref{lem:bgk-1}.
  The trigger string at the end of a proper prefix $\alpha$ of $\alpha\gamma \in \Suf$ would lie in the interior of the phrase $v$.}\label{fig:preffree}
\end{minipage}\hfill
\begin{minipage}[t]{0.485\linewidth}
\centering
\adjustbox{max width=\linewidth}{\includegraphics{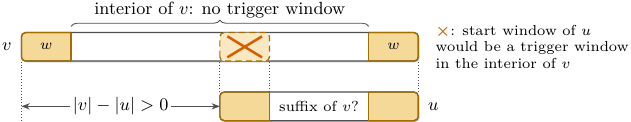}}
\caption{Visualization of \cref{lem:sufffree}.
  The trigger string that starts a shorter phrase $u$ would lie in the interior of $v$.}\label{fig:sufffree}
\end{minipage}
\end{figure}

\noindent If $u$ were a suffix of $v$, the element of $\mathcal{E}'$ that starts $u$ would occur in $v$ at offset $|v| - |u| > 0$. Since $|u| > w$, this occurrence is neither a prefix nor a suffix of $v$, which contradicts the definition of a phrase, see \cref{fig:sufffree}.
We now can sort strings based on their phrase suffixes, because by \cref{lem:bgk-1}, $\alpha$ and $\alpha'$ differ within their first $\min(|\alpha|, |\alpha'|)$ characters.

\begin{lemma}[{\cite[Lemma~2]{BoucherGKLMM19}}]\label{lem:bgk-2}
Let $\alpha\lex \alpha' \in \Suf$. Then $\alpha x \lex \alpha' x'$ for all strings $x$ and $x'$.
\end{lemma}

\noindent We now map all suffixes of length $>w$ of the text to a phrase suffix.
Let $v$ be the phrase that starts with the last occurrence of an element of $\mathcal{E}'$ at or before the start of $x$.
The prefix of $x$ that ends where $v$ ends is a suffix of $v$ longer than $w$, hence in $\Suf$.
It is also the only one.%

\begin{lemma}[{\cite[Lemma~3]{BoucherGKLMM19}}]\label{lem:bgk-3}
Every suffix $x$ of $\#\,T\,\$^w$ with $|x| > w$ has exactly one prefix in
$\Suf$, which we denote by $f(x)\in\Suf$.
\end{lemma}

We can then use the mapping $f(\cdot)$ to sort suffixes of the text based on the lexicographical order of the phrase suffixes.

\begin{lemma}[{\cite[Lemma~4]{BoucherGKLMM19}}]\label{lem:bgk-4}
Let $x$ and $x'$ be suffixes of $\#\,T\,\$^w$ with $|x|, |x'| > w$.  Then
$f(x) \lex f(x')$ implies $x \lex x'$.
\end{lemma}

\noindent This is \cref{lem:bgk-2} applied to $f(x)$ and $f(x')$.
Suffixes $x$ with $f(x) = \alpha$ form a \emph{block} of $\alpha$.
In suffix order, the blocks appear in the order of $\Suf$.
Therefore, sorting $\Suf$ sorts the suffixes of the text up to the order inside each block.
Further copies of $\$$ do not change the order, and $\#\,T\,\$^w$ takes the place of the smallest suffix $\$$.
Hence, the suffixes of $T\$$ can be sorted by their counterparts, which are longer than $w$ and start with phrase suffixes.

\begin{lemma}[{\cite[Lemma~5]{BoucherGKLMM19}}]\label{lem:bgk-5}
The permutation that sorts $T[0..n-1]\$^w, \dots, T[n-1]\$^w, \#\,T\,\$^w$ also
sorts the suffixes $T[0..]\$, \dots, T[n-1]\$, \$$ of $T\$$.
\end{lemma}

\subsection{From the Dictionary and the Parse to the BWT}
\label{sec:pfp-bwt}

Boucher et al.\ map every character $T\$[i]$ to the phrase suffix $\beta(i)$ that starts the next suffix $T\$[i+1..]$.
This is $f(T[i+1..]\$^w)$ for $i < n - 1$, $f(\#\,T\,\$^w)$ for $i = n - 1$, and $f(T\,\$^w)$ for $i = n$, the final $\$$.
By \cref{lem:bgk-4,lem:bgk-5}, the BWT consists of the characters mapped to the phrase suffixes, block by block in the order of $\Suf$.
We consider the order within a block.

\begin{lemma}[{\cite[Lemma~6]{BoucherGKLMM19}}]\label{lem:bgk-6}
  Assume that $\beta$ maps $m$ copies of the character $a$ to $\alpha \in \Suf$ and no other characters.
  Additionally, it maps a total of $o_\alpha$ characters to phrase suffixes smaller than $\alpha$.
  Then, the BWT characters at positions $o_\alpha, \dots, o_\alpha + m - 1$ are copies of $a$.
\end{lemma}

\noindent The block of $\alpha$ occupies exactly these $m$ positions, and their order does not matter if all characters are equal.
We call $o_\alpha$ the \emph{output offset} of $\alpha$ and such a phrase suffix \emph{easy}.
Whether a phrase suffix $\alpha \notin \Dset$ is easy follows from $\Dset$ and the phrase occurrences.
If all phrases of which $\alpha$ is a proper suffix have the same character before $\alpha$, then $\alpha$ is easy.
The run's length is the sum of the number of occurrences.
Otherwise, $\alpha$ is \emph{hard}, and the order inside its block matters.
Since the last $w$ characters of $\alpha$ form the trigger string that starts the next phrase, the suffix of $\Parse$ that starts at the next phrase describes the text after an occurrence of $\alpha$.

\begin{lemma}[{\cite[Lemma~7]{BoucherGKLMM19}}]\label{lem:bgk-7}
  Let $x$ and $x'$ be suffixes of $\#\,T\,\$^w$ that start with occurrences of $\alpha \in \Suf$.
  Also, let $\Parse_x$ and $\Parse_{x'}$ be the suffixes of $\Parse$ that encode the last $w$ characters of these occurrences of $\alpha$ and the remainders of $x$ and $x'$.
  If $x \lex x'$, then $\Parse_x \lex \Parse_{x'}$.
\end{lemma}

\noindent Since $\alpha$ occurs at both $x$ and $x' \ne x$, it does not end with the unique $\$^w$.
Thus, $\Parse_x$ and $\Parse_{x'}$ are nonempty.
Their first difference lies in the first phrase in which $\Parse_x$ and $\Parse_{x'}$ differ.
Since the parse orders phrases by lexicographic rank, the characters of a hard block appear in the BWT in the order of the parse suffixes that follow their occurrences.
If $\alpha$ is a \emph{whole phrase}, i.e., $\alpha \in \Dset$, the character before each of its occurrences is not part of $\alpha$ but lies at distance $w + 1$ from the end of the previous phrase, so the parse provides it.
Hence, the BWT can be computed from $\Dset$ and $\Parse$ alone.

\begin{figure}[tbp]
\centering
\includegraphics{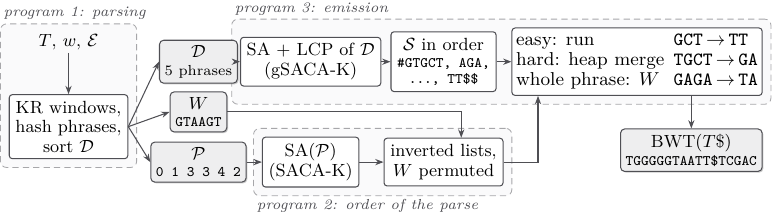}
\caption{Data flow of Big-BWT on \runningexample.}
\label{fig:thm-bgk-1}
\end{figure}

\subsection{The Implementation Big-BWT}\label{sec:pfp-impl}
Big-BWT,\footnote{\url{https://gitlab.com/manzai/Big-BWT}, last accessed 2026-10-22.} the implementation of Boucher et al., consists of three sequential programs that communicate through files, see \cref{fig:thm-bgk-1}.
The first program slides a window of length $w$ over the text and maintains its Karp--Rabin fingerprint with a rolling hash.
When the fingerprint is $0$, the current phrase ends with the window and the next phrase starts with it.
Each finished phrase is looked up by its 64-bit hash in a hash table that stores the distinct phrases with their frequencies.
Afterward, the dictionary is sorted, and the hashes in the parse are replaced by the ranks of their phrases.
For every occurrence of whole phrases, they store, the character at distance $w + 1$ from its end in an array $W[0..|\Parse|-1]$ stores.
The second program computes $\SA(\Parse)$ with SACA-K, writes the \emph{inverted list} of every phrase, i.e., its occurrences in the order of the parse suffixes that follow them (cf.\ \cref{lem:bgk-7}), and permutes $W$ accordingly.
The third program enumerates the phrase suffixes in lexicographic order with the suffix and LCP arrays of the concatenated dictionary (gSACA-K).
For each phrase suffix $\alpha$, they collect the characters that precede $\alpha$ in the phrases that end with it.
If they are all equal, they write a run whose length is the sum of the phrase frequencies (\cref{lem:bgk-6}).
Otherwise, it merges the inverted lists of the phrases with a heap and writes the preceding character of each popped occurrence (\cref{lem:bgk-7}), taken from $W$ if $\alpha$ is a whole phrase.

\section{Our PFP-BWT Variant for the GPU}
\label{sec:algo}
We now show how to replace the hash table, the inverted lists, $W$, and the heap by sorts and scans.
\Cref{tab:mapping} in \cref{app:bigbwt} compares the steps of both algorithms.
Our algorithm keeps the structure of the original, i.e., the parse, the blocks of \cref{sec:pfp-suffixes}, the order of $\Suf$ between blocks, the order of the parse inside hard blocks, and the easy/hard distinction of \cref{lem:bgk-6,lem:bgk-7}.
But all objects of the construction are derived from a single \emph{occurrence array} $\Docc[0..d-1]$ with one entry per phrase occurrence, where $d = |\Parse|$ (\cref{sec:algo-parse}).
$\Docc$ is sorted in three orders: (1) in text order by the parsing (\cref{sec:algo-parse}), (2) lexicographically by phrase for the dictionary and the parse (\cref{sec:algo-dict}), and (3) lexicographically by reversed phrase for the phrase suffixes and their occurrences (\cref{sec:algo-phases,sec:algo-emit}).
\Cref{fig:gpu-flow} shows the resulting data flow.

\begin{figure}[tbp]
\centering
\includegraphics{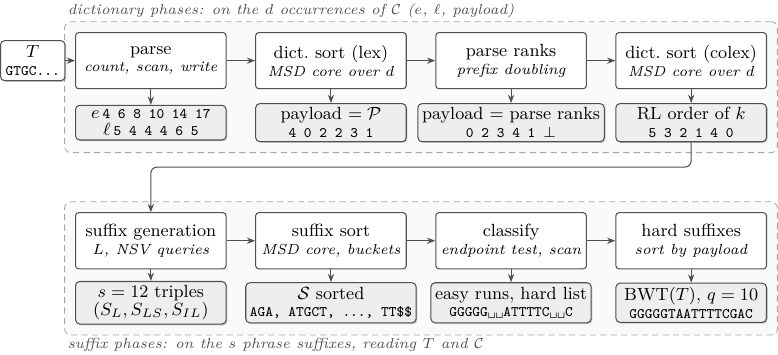}
\caption{Data flow of \gpupfpbwt in the configuration \code{cub6} on \runningexample.
  The top row shows the \emph{dictionary phases}, the bottom row the \emph{suffix phases}. NSV denotes next-smaller-value queries. In \Cref{sec:memory,sec:refinements}, we define buckets, the MSD core, and the configurations.
  $\bot$ marks the last occurrence.
  The gaps (\textvisiblespace) are filled by the hard suffixes.
}\label{fig:gpu-flow}
\end{figure}

\subsection{Parsing}\label{sec:algo-parse}

We parse $T\,\$^w$ without prepending $\#$.
Instead, $T[0..w-1]$ is a trigger string.
This keeps the text unchanged and yields the BWT directly in the convention of \libcubwt, i.e., the character that precedes position $0$ is $T[n-1]$. %
Let $w$ be the window length and $\mathcal{E}$ a set of trigger strings of length $w$.
For $1 \le i \le n - w$, position $i$ is a \emph{trigger position} if $T[i..i+w-1] \in \mathcal{E}$ or $T[i..i+w-1] = T[0..w-1]$.
Let $e_0 < \dots < e_{d-2}$ be the ends $i + w - 1$ of the windows at trigger positions $i$, $e_{d-1} = n + w - 1$ the end of $T\,\$^w$, and $e_{-1} = w - 1$.
Phrase $k$ is $T\,\$^w[e_{k-1} - w + 1..e_k]$, of length $\ell_k = e_k - e_{k-1} + w$.
The \emph{occurrence array} $\Docc[0..d-1]$ stores phrase $k$ as the entry $\Docc[k] = (e_k, \ell_k, k)$, i.e., its end, its length, and a \emph{payload} that initially is its text-order index.

The \emph{parse} $\Parse$, the \emph{dictionary} $\Dset$, and the phrase suffixes $\Suf$ are defined as in \cref{sec:pfp}.
We write $\suf(k, j)$ for the suffix of length $j$ of phrase $k$ and $\pre(k, j) = T[(e_k - j) \bmod n]$ for the character that precedes it. For $j = \ell_k$, this is the character before phrase $k$.
Phrase $0$ starts with the trigger string $T[0..w-1]$.
Every other phrase is delimited as in the original.
Hence, \cref{lem:bgk-1,lem:sufffree,lem:bgk-3,lem:bgk-4} hold for our parse, and \cref{lem:parse-rank} below adapts \cref{lem:bgk-7}.%

\subsection{Dictionary, Parse, and Parse Ranks}\label{sec:algo-dict}

We sort a copy of $\Docc$ by phrase content read left to right, breaking ties by length (\emph{LR order}).
We mark every occurrence whose phrase differs from that of its predecessor and propagate the group starts by pointer jumping.
Each group of equal phrases then receives its lexicographic rank, i.e., its identifier in $\Parse$.
Each occurrence of the copy writes its identifier to the payload of $\Docc[k]$, where $k$ is its text-order index.
Now, reading the payloads in text order yields $\Parse$.
The dictionary is never materialized and the frequency of a phrase is the size of its group.

We compute $\SA(\Parse)$ by prefix doubling on rank pairs.
Here, the payload is overwritten as follows.
Occurrence $k$ gets its \emph{parse rank}, i.e., the rank of the parse suffix starting at $k + 1$.
The last occurrence gets the sentinel rank $\bot$, which is smaller than all other ranks.
Ordering occurrences by their payload thus orders them by the text that follows them.

\begin{lemma}\label{lem:parse-rank}
Let $\alpha \in \Suf$, and let $k \ne k'$ be occurrences whose phrases end with $\alpha$.
The suffix of $T\,\$^w$ that starts with the occurrence of $\alpha$ at the end of phrase $k$ is smaller than the one for phrase $k'$ if and only if the payload of $k$ is smaller than that of $k'$.
\end{lemma}

\begin{proof}
For $0 \le i < d$, let $y_i = T\,\$^w[e_{i-1} - w + 1..]$ be the suffix of $T\,\$^w$ that starts with phrase $i$.
Neither $k$ nor $k'$ is the last occurrence $d - 1$, i.e., the only phrase that contains $\$$.
If it ended with $\alpha$, then $\alpha$ would be $\$^w$ preceded by at least one character.
Such a string occurs only once in $T\,\$^w$
Hence, the payloads of $k$ and $k'$ are the ranks of the nonempty parse suffixes $\Parse[k+1..]$ and $\Parse[k'+1..]$ among all suffixes of $\Parse$.

Now, let $j = |\alpha|$.
Phrase $k + 1$ starts $w$ characters before the end of phrase $k$.
The suffix that starts with the occurrence of $\alpha$ in $k$ and $k'$ are therefore $\alpha[0..j-w-1]\,y_{k+1}$ and $\alpha[0..j-w-1]\,y_{k'+1}$.
Both differ in $y_{k+1}$ and $y_{k'+1}$ do.
We now show that $y_i \lex y_{i'}$ if and only if $\Parse[i..] \lex \Parse[i'..]$ for $i \ne i'$.
Every phrase is an element of $\Suf$, so the phrases are prefix-free (\cref{lem:bgk-1}).
Let $\delta$ be the first index with $\Parse[i+\delta] \ne \Parse[i'+\delta]$.
The text $y_i$ consists of phrase $i$ followed by every later phrase without its first $w$ characters.
The first distinct phrases $u = \Parse[i+\delta]$ and $u' = \Parse[i'+\delta]$ differ at a position inside both.
Hence, $y_i \lex y_{i'}$ holds if and only if $u \lex u'$, i.e., if and only if $\Parse[i..] \lex \Parse[i'..]$.
With $i = k + 1$ and $i' = k' + 1$, $\Parse[k+1..]$ and $\Parse[k'+1..]$ are ordered correctly by the order of their payloads.
\end{proof}

\subsection{Phrase Suffixes and Their Occurrence Sets}\label{sec:algo-phases}

A second sort of $\Docc$ yields the phrase suffixes and the occurrences that end with each of them.
We sort $\Docc$ by phrase read right to left from $e_k$, breaking ties by length and then by the preceding character $\pre(k, \ell_k)$.
In this \emph{RL order}, occurrences that share a phrase suffix are adjacent. %
Now, $k$ denotes an RL position, and $e_k$, $\ell_k$, $\suf(k, j)$, and $\pre(k, j)$ refer to the occurrence at this position.
Let $L[k]$ be the length of the longest common \emph{suffix} of the phrases at RL positions $k$ and $k - 1$.
The suffixes $\suf(k, j)$ with $j \le L[k]$ already end the phrase at position $k - 1$, and those with $j \le w$ are not in $\Suf$.
Position $k$ therefore \emph{generates} only the remaining ones, i.e., $\suf(k, j)$ for $\max(w, L[k]) < j \le \ell_k$.
\begin{lemma}\label{lem:interval}
  Let $\alpha \in \Suf$ with $|\alpha| = j$.
  In the RL order, the occurrences $k$ with $\suf(k, j) = \alpha$ form a contiguous interval.
  The leftmost $k$ is the unique position in it with $L[k] < j$.
\end{lemma}

\begin{proof}
  The RL order sorts occurrences by their reversed phrases $\rev{v}$.
  Occurrences with a common prefix $\rev{\alpha}$ are contiguous and have $\suf(k, j) = \alpha$.
  Now, the common-prefix length with the predecessor is $<j$ exactly at the start of a run sharing a prefix of length $\ge j$.
\end{proof}

By \cref{lem:interval}, this generation produces each element of $\Suf$ exactly once, at its leftmost occurrence.
Each of the $s = |\Suf|$ elements is represented by a triple containing its leftmost RL occurrence $S_L = \lb$, its length $S_{LS} = j$, and its \emph{interval length} $S_{IL} = \rb - \lb$, where $[\lb, \rb]$ is its occurrence interval.
Here, $\rb + 1$ is the next position with $L < j$, which can be found by a next-smaller-value (NSV) query on $L$.
Finally, the triples are sorted by the strings $\suf(S_L, S_{LS})$ they represent.

\subsection{Generating the BWT}\label{sec:algo-emit}

A scan over the interval sizes $S_{IL} + 1$ of the sorted phrase suffixes yields the output offset of every phrase suffix.
For a phrase suffix of length $j$ with interval $[\lb, \rb]$, the classification in easy and hard is a comparison of $\pre(\lb, j)$ and $\pre(\rb, j)$.

\begin{lemma}\label{lem:monotone}
  Let $[\lb, \rb]$ be the interval of $\alpha \in \Suf$ with $|\alpha| = j$.
  Then $\pre(\lb, j) \le \pre(\lb+1, j) \le \dots \le \pre(\rb, j)$.
  Thus, all occurrences of $\alpha$ are preceded by the same character if and only if $\pre(\lb, j) = \pre(\rb, j)$.
\end{lemma}
\begin{proof}
  Let $v_k$ be the phrase at RL position $k$.
  By \cref{lem:sufffree}, either $j < \ell_k$ for every $k \in [\lb, \rb]$, or $j = \ell_k$ for every such $k$ and all these occurrences carry the same phrase.
In the first case, $\pre(k, j)$ is the character that follows $\rev{\alpha}$ in $\rev{v_k}$. %
Since the RL order is lexicographic on reversed phrases, the sequence is nondecreasing.
In the second case, all reversed phrases and all lengths are the same.
Hence, they are sorted by their last tie-break, which is exactly $\pre(k, \ell_k) = \pre(k, j)$.
\end{proof}

If they are equal, the phrase suffix is easy and contributes a run of $\rb - \lb + 1$ copies of that character.
Otherwise, it is hard and its occurrences are sorted by the parse rank in their payload, and their preceding characters are written in this order.
Whole phrases need no special case, since $\pre(k, \ell_k)$ is read from the text like every other preceding character.

\begin{theorem}\label{prop:emit}
Processing the elements of $\Suf$ in lexicographic order and outputting, for each element $\alpha$, the preceding characters of its interval $[\lb,\rb]$, as a run when $\pre(\lb,|\alpha|)=\pre(\rb,|\alpha|)$ and ordered by the parse rank in the payload otherwise, produces $\BWT(T)$.
\end{theorem}
\begin{proof}
  The suffixes of $T\,\$^w$ that start in $T$ are partitioned into the blocks of the phrase suffixes, ordered by the lexicographic order of $\Suf$ (\cref{lem:bgk-3,lem:bgk-4}).
  The block of $\alpha$ consists of the $\rb - \lb + 1$ occurrences in its interval (\cref{lem:interval}).
  Hence, the interval sizes sum to $n$, and their prefix sums are the correct block starts.
  Inside a block, the order is that of the parse ranks in the payload (\cref{lem:parse-rank}).
  When all preceding characters coincide, the order inside the block is irrelevant (\cref{lem:bgk-6}), which is exactly the easy case of \cref{lem:monotone}.
  Finally, since we do not prepend $\#$ and read $\pre(k, j)$ cyclically, the suffix starting at position $0$ is preceded by $T[n-1]$, and no $\$$ is output.%
\end{proof}

\subparagraph*{The Primary Index.}\label{sec:algo-primary}

As a drop-in replacement for \libcubwt, our construction must also return the primary index $q$ with $\SA[q] = 0$.
We derive $q$ after the classification, while the sorted triples and the parse ranks are still available.
Let $k_0$ be the RL position of the occurrence with the smallest end, i.e., of phrase~$0$, and let $\ell_0$ be its length.
The suffix starting at $0$ begins with the element of $\Suf$ of length $\ell_0$ whose interval contains $k_0$, see \Cref{fig:primary}.

\begin{lemma}\label{lem:primary}
Let $\alpha_0$ be the element of $\Suf$ that is a prefix of the suffix starting at text position $0$,
$[\lb, \rb]$ its interval, and $o_{\alpha_0}$ its output offset.  Then
$q = o_{\alpha_0} + |\{\, k \in [\lb, \rb] : \rank(k) < \rank(k_0) \,\}|$, where $\rank$ is
the parse rank in the payload.
\end{lemma}
\begin{proof}
  The positions of $\BWT$ that belong to the occurrences of $\alpha_0$ form the interval $[o_{\alpha_0}, o_{\alpha_0} + \rb - \lb]$ (\cref{prop:emit}).
  Within it, the occurrences appear in increasing order of their parse rank.
  This holds both for easy intervals and for hard ones.
  The position of occurrence $k_0$ is therefore $o_{\alpha_0}$ plus the number of occurrences in the interval with a smaller rank.
\end{proof}

\begin{figure}[tbp]
\centering
\includegraphics{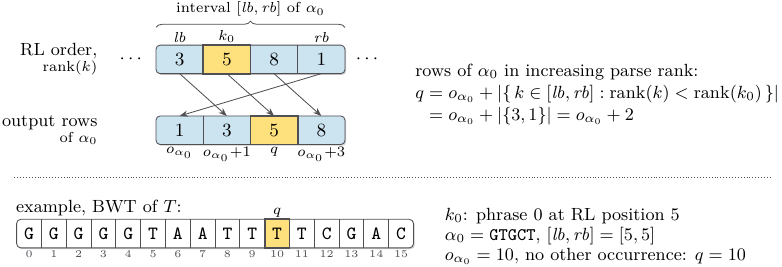}
\caption{The primary index. The occurrences in $[\lb, \rb]$ map to the BWT positions $[o_{\alpha_0}, o_{\alpha_0} + \rb - \lb]$ in increasing parse rank, here $q = o_{\alpha_0} + 2$ (top).
  In \runningexample, $\alpha_0 = \texttt{GTGCT}$ with $[\lb, \rb] = [5, 5]$ and $o_{\alpha_0} = 10$, so $q = 10$ (bottom).}\label{fig:primary}
\end{figure}

\section{Engineering PFP-based BWT Construction on the GPU}
\label{sec:engineering}
A simple direct implementation of the algorithm needs far more memory than \libcubwt.
We therefore show how to minimize the peak resident device memory in practice.
The input text and the output take $n$ bytes each, i.e., we only consider byte alphabets.
We also limit the input to at most 4\,GiB, as this covers GPUs with up to 40\,GiB memory but allows us to use 32-bit indices.
The occurrence array $\Docc$ takes $12d$ bytes for its $d = |\Parse|$ entries of three 4-byte values on all realistic inputs.
Each phrase suffix is stored as a $10$-byte triple.
Later, we show that the suffix sort can processes \emph{buckets} of suffixes with the same first character separately.

\subsection{Improving the Memory Footprint} \label{sec:memory}
In the baseline (opt~0), the suffix sort is a character-by-character refinement.
For a suffix $x$ with phrase end $e(x)$ and length $\ell(x)$, the padded character sequence is $c_t(x) = T[e(x) - \ell(x) + 1 + t]$ for $t < \ell(x)$ and $0$ otherwise.
The order $\lex$ is the lexicographic order on padded sequences, with ties broken by length.
For each unfinished suffix we need a position and a rank.
Maximal runs of equal rank form \emph{rank groups}.
In the $t$th iteration, we refine the groups by $c_t$, sorting \emph{$B$-blocks} of $B = 512$ elements with $H = 256$ histogram slots each.
A group that fits into one block is finished by a bitonic sort in shared memory.
After every iteration, we discard~\cite{DementievKMS08,BingmannGK18,FischerK19,FlickA15} the finished suffixes.
The refinement maintains two invariants.
At every compaction point, every unfinished suffix $x$ satisfies the following.
\begin{enumerate}[nosep,label=(I\arabic*)]
\item\label{inv:rank} $\rank(x) = |\{\, y : y \lex x \text{ and }
  y \notin \mathrm{group}(x) \,\}|$, counted over all $s$ suffixes, finished or not.
  Thus, the rank is constant on a group, and once a group is ordered internally, its member $i$ (counted from $0$) has the final position
  $\rank + i$.
\item\label{inv:agree} After iteration $t$, all members of a rank group
  agree on $c_0, \dots, c_t$.
\end{enumerate}

Let $s_t$ be the number of unfinished suffixes at the start of iteration $t$ and $b_t$ the number of $B$-blocks that cover the rank groups.
The baseline keeps thirteen arrays of length $s$, i.e., $56$ bytes per suffix consisting of the triples ($10$ bytes) with cached phrase ends $e(x)$, the \emph{suffix ends} ($4$ bytes), the sorting state ($16$ bytes) with two cached 8-byte words of upcoming characters ($16$ bytes), and output arrays for the sorted triples ($10$ bytes). They are allocated before the memory of the interval computation is freed.
Three auxiliary buffers hold $H$ histogram slots of $4$ bytes per block, so each needs $4 \max(H b_t, s_t)$ bytes in iteration $t$.
A group $G$ gets $\lceil |G| / B \rceil \ge 1$ blocks. %
Since the buffers only grow they may need up to $12 \max_t \max(H b_t, s_t)$ bytes, see \cref{fig:baseline}. %

\begin{figure}[tbp]
\centering
\includegraphics{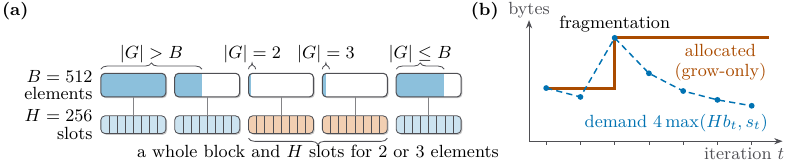}
\caption{Memory profile of the baseline.
  (a)~Groups of $2$ or $3$ elements occupy a whole block and its histogram slots.
  (b)~The $4 \max(H b_t, s_t)$ byte peak.}\label{fig:baseline}
\end{figure}

\subparagraph*{Allocation hygiene}
To improve this, we allocate the $42$ bytes per element of sorting state and output arrays only after the $\Theta(s + d)$-byte memory of the interval computation is freed.
To this end, we obtains every Thrust temporary from one cached allocator. %

\subparagraph*{Finishing small groups early.}
Next, we finalize small groups before they receive a block.
After each discarding, one $64$-thread block per group sorts every group of at most $B$ elements in shared memory and writes it to the output arrays.
A second discarding step removes these groups \emph{before} the blocks and histogram slots of the next iteration are created.
Every group that still receives blocks then has more than $B = 2H$ elements.
Hence, $H b_t < s_t$, and the three auxiliary buffers take at most $12 \max_t s_t \le 12 s$ bytes in total, see \Cref{fig:blocks}. %

\begin{figure}[tbp]
\centering
\includegraphics{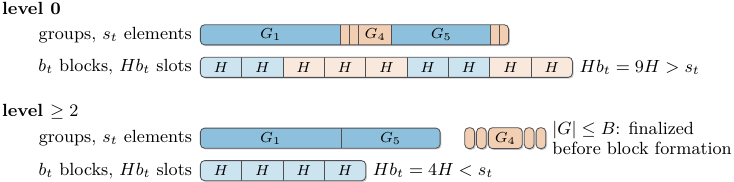}
\caption{Block formation at level~$0$ and when finishing small group early ($H = B/2$).
}\label{fig:blocks}
\end{figure}

\subparagraph*{Recomputations.}
Next, we drop the two cached 8-byte words of upcoming characters, as they can be recomputed from $T$ and the suffix end.

\begin{lemma}\label{lem:load}
For every suffix $x$ and iteration $t$, the
$[\,i \le e(x)\,] \cdot T[i]$ with $i = e(x) - \ell(x) + 1 + t$ and
$e(x) = \Docc[S_L(x)]$ contains exactly the value that the baseline extracts from its cached 8-byte word, including all padding cases.
\end{lemma}

\begin{proof}
  With the word base $t_0 = t - (t \bmod 8)$, the baseline loads the eight bytes at the positions $e - \ell + 1 + t_0 + i$ for $i = 0, \dots, 7$, replaces positions beyond $e$ by $0$, packs byte $i$ at bit offset $8i$, and extracts byte $t \bmod 8$.  This is the value at position $e - \ell + 1 + t$.
  The cached suffix end equals $\Docc[S_L(x)]$ and is never modified.
\end{proof}

\subparagraph*{MSD bucketing by the first character.}
In one most-significant-digit (MSD) step, we now partition the suffixes into buckets $\mathcal{B}_c = \{x : c_0(x) = c\}$ with one counting pass and one scatter.
Sorting now happens one bucket at a time, starting at iteration $t = 1$.
The sorting state is allocated once, for the $\bmax = \max_c |\mathcal{B}_c|$ entries of the largest bucket.

\begin{lemma}[Bucket decomposition]\label{lem:bucket}
The order $\lex$ orders $\mathcal{B}_0 \lex \mathcal{B}_1 \lex \dots \lex \mathcal{B}_{\sigma+1}$
blockwise. Within each bucket, the loop started with
$\rank(x) \equiv \mathrm{base}_c := \sum_{c' < c} |\mathcal{B}_{c'}|$ and $t = 1$
computes the correct final positions
$[\mathrm{base}_c, \mathrm{base}_c + |\mathcal{B}_c|)$.
\end{lemma}
\begin{proof}
  The character $c_0$ is never a padding character, since $\ell(x) \ge 1$.
  Therefore, the buckets are contiguous in the order $\lex$.
  The initialization makes the whole bucket one rank group with rank $\mathrm{base}_c$, i.e.,  invariant \ref{inv:rank} for the group structure with one group per bucket.
  This is because the suffixes outside the group that are smaller are in smaller buckets.
  Invariant \ref{inv:agree} holds for $t = 0$ by construction.
  For $t>0$, elements of different groups are never compared, making the buckets are independent.
  By invariant \ref{inv:rank}, the bucket $c$ land in $[\mathrm{base}_c, \mathrm{base}_c + |\mathcal{B}_c|)$.
  These intervals are disjoint and cover $[0, s)$.
\end{proof}

Now, only the triples with their cached suffix ends, the output arrays, and the \emph{bucket permutation} computed by the scatter (freed after the bucket loop) have $s$ entries.
The four sorting-state arrays and the auxiliaries have $\bmax$ entries.
Together with $\Docc$, the input, the BWT output, and $20$-byte block descriptors for $b \le \lceil \bmax / B \rceil + \bmax / H$ blocks, the device memory resident during the suffix sort is at most
\begin{equation}\label{eq:level4}
  \underbrace{14s}_{\text{triples}} + \underbrace{10s}_{\text{output}}
  + \underbrace{4s}_{\text{bucket permutation}}
  + \underbrace{16 \bmax}_{\text{sorting state}}
  + \underbrace{12 \bmax}_{\text{auxiliaries}}
  + 12d + 2n + O(\bmax / B)\ \text{bytes}.
\end{equation}
At the beginning, the bound is $56s + 12 \max_t \max(H b_t, s_t) + 12d + 2n$ bytes, with $\max_t H b_t \in [s/2, H s]$, see \Cref{fig:thm-main}.
Let $\theta$ be the most frequent first character's relative frequency.
Then $\bmax = \theta s$ and the sum reduces from at least $44s$ bytes %
to $28 \theta s$ bytes.
On the Pizza\&Chili text \texttt{english} ($\theta \approx 0.17$) we could more than halve the peak, see \cref{tab:texts} in \cref{sec:eval}.

\begin{figure}[tbp]
\centering
\adjustbox{max width=\linewidth}{\includegraphics{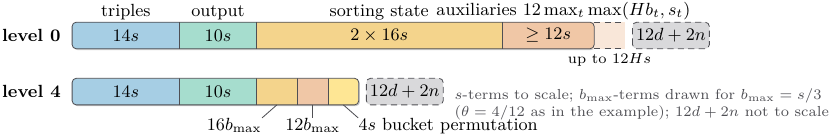}}
\caption{Per-suffix device memory of the suffix sort at level~$0$ and level~$4$ (\cref{eq:level4}).
  The $s$-terms are drawn to scale, the $\bmax$-terms for $\theta = 1/3$ as in \runningexample.
  The term $12d + 2n$ is not to scale.}\label{fig:thm-main}
\end{figure}

\subsection{Performance improvements}\label{sec:refinements}
In the previous section, we up to halved the memory requirements of our implementation.
Unfortunately, this also reduced the performance significantly.
Now, we show, how to improve the throughput while also slightly improving the memory requriements.

\subsubsection{In-Place Bucketed Output}\label{sec:ref-level56}

Up till now, we have keept two copies of the triples, i.e., the unsorted arrays that the refinement reads and the output arrays to which it writes finished suffixes.
We now remove the second copy.
First, a stable least-significant-digit pass (a histogram, a scan, and a scatter through one 4-byte and one 2-byte temporary) moves the triples of each first-character bucket into a contiguous slice.
Second, consecutive buckets form chunks of at most $\max(\bmax, s/K)$ suffixes ($K = 256$ by default).
A finished suffix records its final rank in a per-chunk permutation instead of copying its triple.
Once a chunk is sorted, the permutation is applied in place to its slice. %
Third, the easy/hard classification replaces its three $4s$-byte prefix-sum auxiliaries by one byte per suffix.

On skewed inputs, the largest first-character bucket dominates the chunk size.
Therefore, we presort by the first \emph{two} characters ($65\,536$ buckets, via a second stable pass).
This bounds $\bmax$ by the frequency of the most frequent bigram.
Now, the bucketing pass holds the triples ($10s$ bytes), the two temporaries ($6s$), $\Docc$, the input, and the output, i.e., $16s + 12d + 2n$ bytes.
During the sort, the triples, the cached suffix ends ($4s$), and the sorting state of the chunk are required.
The sorting state takes $28$ bytes per chunk slot for the refinement of \cref{sec:memory} (the \emph{refinement core}) and $44$ for the MSD sort introduced next in \cref{sec:ref-msd}. %
With the MSD sort, the memory during the sorting is $14s + 12d + 2n + 44\max(\bmax, s/K)$ bytes.
For $K = 256$ and chunks of $\max(\bmax, s/K)$ suffixes, the bucketing is the peak of the suffix phases unless a single bucket exceeds $s/22$, i.e., $4.5\%$ of $\Suf$, see \Cref{fig:level6}.

\begin{figure}[tbp]
\centering
\includegraphics{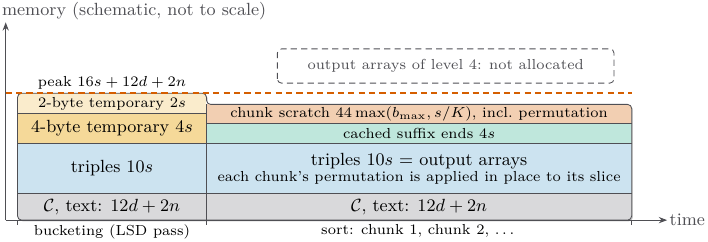}
\caption{Memory usage of the suffix phases during the bucketing pass and during the chunk-by-chunk sort as described in \cref{sec:ref-level56}.}\label{fig:level6}
\end{figure}

\subsection{MSD Radix Sort}\label{sec:ref-msd}

Now, we let the \emph{MSD radix sort} compare eight characters per round with two stable radix-sort passes of CUB.
We use it for both \emph{dictionary sorts}, i.e., the LR and the RL sort of $\Docc$ and also for the suffix chunks.
Every active element carries a rank and a character offset.
A round sorts the elements by the 8-character key at their offset (padded with $0$) and, stably, by rank, so that equal keys within a rank group become adjacent.
Two max-scans mark group starts and \emph{runs} of equal keys.
Every run receives the rank ``group rank $+$ run start $-$ group start''.
In the RL dictionary sort, the key of an exhausted phrase is its preceding character, which we can use as tie-break without a separate pass.
This needs $44$ bytes of space per chunk.

A round costs a fixed part (eight kernel launches, three scans, and one host synchronization) plus a part proportional to the number of active elements.
A chunk needs $\lceil \lambda / 8 \rceil$ rounds, where $\lambda$ is the length of the longest common prefix of two of its suffixes.

A small input has not to be split into $K$ chunks.
We therefore set the chunk size to $\max(\bmax, s/K, \min(2^{22}, s))$ instead of $\max(\bmax, s/K)$, at a negligible $44 \cdot 2^{22}$ bytes of additional space.

Suffixes with long common prefixes spend one round per eight shared characters.
After the discarding step, every remaining suffix therefore compares the 8-byte words that follow its key with those of the first element of its rank group, up to the first difference or $32$ words.
The group minimum $\mu$ of these counts, computed with an \code{atomicMin} at the group's first index, advances the offset of every member by $8\mu$ further characters, see \Cref{fig:skip}.

\begin{lemma}\label{lem:skip}
  Let $G$ be a rank group after the compaction of a round, with all members at offset $\tau$.
  Let $\mu$ be the minimum over $x \in G$ of the number of consecutive 8-byte words of $x$ from offset $\tau + 8$ on that equal the corresponding words of the first member $g$ of $G$.
  Advancing the offset of every member of $G$ to $\tau + 8 + 8\mu$ does not change the sorted order computed by the subsequent rounds.
\end{lemma}

\begin{proof}[Proof of \cref{lem:skip}]
By invariant \ref{inv:agree}, which the MSD sort maintains like before.
For $x, y \in G$ and any word index $i < \mu$, the $i$th words
of $x$ and of $y$ both equal that of $g$.
Hence, $x$ and $y$ agree on all characters in $[\tau + 8, \tau + 8 + 8\mu)$.
The order of $G$ under $\lex$ is therefore decided by the characters from $\tau + 8 + 8\mu$ on and the length rule.
\end{proof}

\begin{figure}[tbp]
\centering
\includegraphics{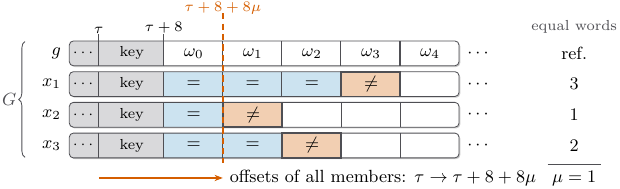}
\caption{LCP skipping in a rank group $G$ at offset $\tau$ (schematic).
  Equal words are shaded and marked $=$, up to the first differing word ($\ne$).
  With the minimum count $\mu = 1$, all members advance to the offset $\tau + 8 + 8\mu$.}\label{fig:skip}
\end{figure}

\subsection{Saving More Space}\label{sec:ref-space}

Before, the $16s + 12d + 2n$ bytes, $4s$ bytes are taken by the interval lengths $S_{IL}$ and the temporary through which they are permuted.
On repetitive inputs, the dictionary sorts determine the peak, c.f., \cref{sec:eval}.
We now reduce this even further.
To this end, we first compute the interval lengths only after the suffix sort, as described in  \cref{sec:algo-phases}.
In \emph{phrase-sequential} order, i.e., by RL position and then it, by increasing suffix length.
This lowers the bucketing cost to $12s + 12d + 2n$ bytes.
After the sort, $10s + 16d + 2n$ bytes are used.%
Second, before, we run both dictionary sorts over all $d$ entries at once, with $44d$ bytes of memory next to the $36d$ bytes of $\Docc$ and its copies.
Now, we bucket the entries by their first (LR sort) or last (RL sort) two characters with a stable 16-bit radix sort instead.
We sort consecutive buckets in $16$ chunks, with initial ranks set to the bucket starts and the first round at offset~$2$.
Equal phrases share a bucket and keep their relative order.%

\subsection{Prefix-Doubling}\label{sec:ref-doubling}

Even with LCP skipping, the MSD sort needs up to $\lceil \lambda / 8 \rceil$ rounds for a common prefix of length $\lambda$, each of which pays the fixed cost.
During the suffix sort replaced by prefix doubling~\cite{ManberM93,LarssonS07}, which needs $O(\log \lambda)$ rounds.
For a suffix and a distance $h$, doubling needs the rank of the suffix that starts $h$ characters later.
We sort a text instead of the set $\Suf$.
The \emph{dictionary text} $\Dict$ is the concatenation of the phrases that generate at least one element of $\Suf$, in RL order followed by a separator $\texttt{0}$.
Every element of $\Suf$ is a suffix of $\Dict$ at a position that follows from its phrase and length.

\begin{lemma}\label{lem:dict-order}
For two elements $\alpha \ne \alpha'$ of $\Suf$ at positions $i$ and $i'$ of $\Dict$,
$\alpha \lex \alpha'$ if and only if $\Dict[i..] \lex \Dict[i'..]$.
\end{lemma}
\begin{proof}
Since $\Suf$ is prefix-free (\cref{lem:bgk-1}) and contains every distinct phrase suffix once, $\alpha$ and $\alpha'$ differ at a position within $\min(|\alpha|, |\alpha'|)$, where both orders compare the same characters.
\end{proof}

Hence, any suffix array construction for $\Dict$ yields the order of $\Suf$.
Round~0 uses 64-bit keys that pack $\kappa = \lfloor 64 / \nu \rfloor$ characters of $\nu$ bits, where $\nu$ is the width of the largest character of $\Dict$.%
It computes the ranks after the first $\kappa$ characters.%
The rank of a run of equal keys is its first index.
The positions in runs of size at least two form the active list.
In the doubling rounds $h = \kappa, 2\kappa, 4\kappa, \dots$, the active positions are sorted by the 64-bit key $(\rank[i], \rank[i + h])$, with rank $0$ past the end.
Within a group of equal old rank, every run of equal keys receives the rank ``old rank $+$ run start $-$ group start''.
Runs of size one are resolved before the ranks are written back.

\section{Experimental Evaluation}\label{sec:eval}
In this section, we compare our implementation the improvements presented in \cref{sec:engineering} with the current state-of-the-art GPU-based BWT construction algorithm \libcubwt.
We denote the baseline implementation as \emph{op0}.
The configurations \emph{opt1} to \emph{opt4} add the allocation hygiene, finish small groups early, recomput instead of cache, and use MSD bucketing, respectively, see \cref{sec:memory}.
Configurations \emph{opt5} and \emph{opt6} use the in-place bucketed output without and with presorting by two characters (\cref{sec:ref-level56}).
When we use the more sophisticated MSD radix sort (\cref{sec:ref-msd}) with opt5 and opt5, this is denoted by \emph{cub5} and \emph{cub6}, resp.
We use \emph{fast} to denote a version of cub6 that sacrifices some space savings for more performance.
The space-efficient configuration (\cref{sec:ref-space}) is denoted by \emph{space} and the prefix doubling approach (\cref{sec:ref-doubling}) for the suffix sorting by \emph{dbl}.

\subparagraph*{Setup.}
All experiments run on an NVIDIA GeForce RTX 3090 under CUDA 12.6. 
The reported times cover the whole construction from host input to host output.
Every reported time is the median of three runs and validated against a main memory BWT construction algorithm~\cite{libsais}.
The inputs are the texts of the Pizza\&Chili corpus~\cite{PizzaChili}, see \cref{tab:texts} for more details.
To find the largest input that a configuration handles, we run it on prefixes of every input, from $64$\,MiB in steps of $25\%$ up to the whole file.

\begin{table}[t]
\caption{The inputs with their size, alphabet size $\sigma$, number of BWT runs
$r$, and $n / r$.  The left part lists the non-repetitive texts and the right
part the repetitive ones.}\label{tab:texts}
\centering\small
\setlength{\tabcolsep}{3.5pt}
\begin{tabular}{@{}lrrrr@{\hspace{12pt}}lrrrr@{}}
\toprule
text & MiB & $\sigma$ & $r$ & $n / r$ & text & MiB & $\sigma$ & $r$ & $n / r$ \\
\midrule
\texttt{dblp.xml} & 282.4 & 97 & 41\,037\,553 & 7.2 & \texttt{cere} & 439.9 & 5 & 11\,574\,639 & 39.9 \\
\texttt{dna} & 385.2 & 16 & 243\,492\,866 & 1.7 & \texttt{coreutils} & 195.8 & 236 & 4\,684\,458 & 43.8 \\
\texttt{english} & 2\,108.0 & 239 & 658\,301\,008 & 3.4 & \texttt{einstein.de.txt} & 88.5 & 117 & 101\,369 & 915.1 \\
\texttt{pitches} & 53.2 & 133 & 23\,429\,976 & 2.4 & \texttt{einstein.en.txt} & 446.0 & 139 & 290\,237 & 1\,611.2 \\
\texttt{proteins} & 1\,129.2 & 27 & 441\,858\,468 & 2.7 & \texttt{Escherichia\_Coli} & 107.5 & 15 & 15\,044\,485 & 7.5 \\
\texttt{sources} & 201.1 & 230 & 47\,896\,806 & 4.4 & \texttt{influenza} & 147.6 & 15 & 3\,022\,820 & 51.2 \\
& & & & & \texttt{kernel} & 246.0 & 160 & 2\,791\,366 & 92.4 \\
& & & & & \texttt{para} & 409.4 & 5 & 15\,636\,738 & 27.4 \\
& & & & & \texttt{world\_leaders} & 44.8 & 89 & 573\,485 & 81.9 \\
\bottomrule
\end{tabular}
\end{table}

\subparagraph*{Largest inputs.}
Our our best configurations process all 15 inputs completely.
\libcubwt process $14$ of the $15$ inputs and failed on \texttt{english}.
Here, \libcubwt handles prefixes of up to $1\,164$\,MiB, whereas we  process the whole $2\,108$\,MiB, i.e., $1.8$ times as much.
\libcubwt needs $20.4$ bytes per character on every input and size, \code{cub6} $6.1$ to $17.0$, and \code{space} $5.9$ to $14.6$ at the largest size they handle.

\begin{figure}[tbp]
\centering
\includegraphics[width=0.96\textwidth]{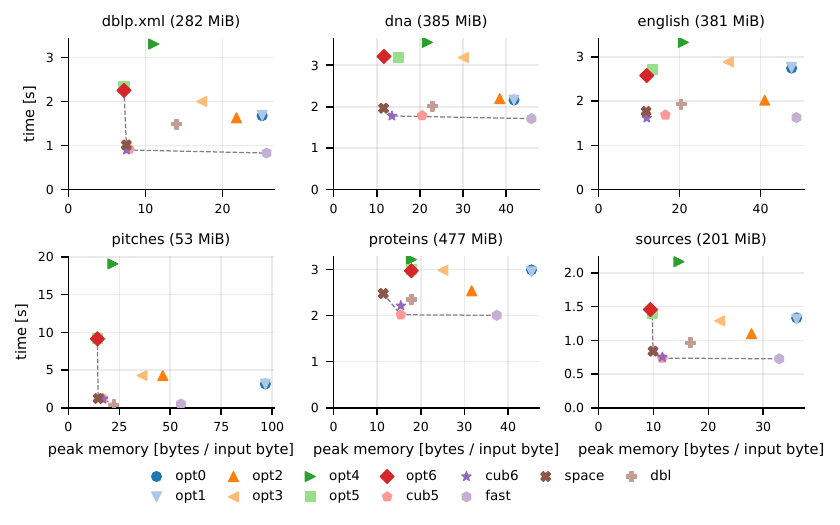}\\[1ex]
\includegraphics[width=0.96\textwidth]{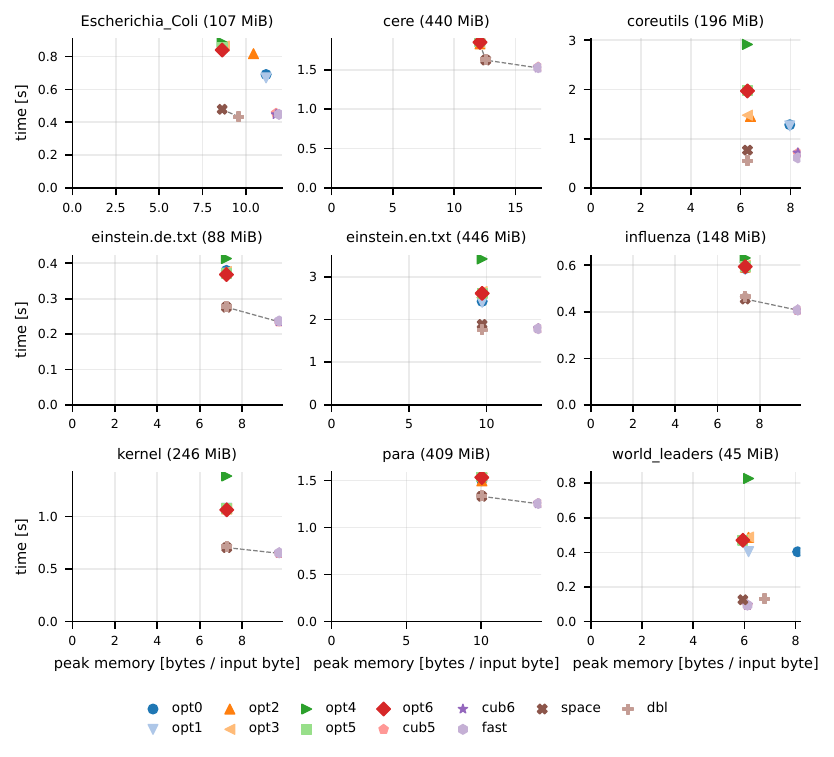}
\caption{Pareto front (dashed) of peak memory per input character (x-axis) and time (y-axis) of all configurations at the largest size that all of them handle.}\label{fig:spacetime-all}
\end{figure}

\subparagraph*{Comparison with \libcubwt.}
\Cref{tab:libcubwt} compares \code{fast}, \code{cub6}, \code{space}, and \code{dbl} with \libcubwt at the largest size that all five handle.
The space-time trade-off is visualized in \cref{fig:spacetime-all}.
\Cref{fig:throughput,fig:spacetime-selected} show the throughput over all sizes and the space-time points.
On the non-repetitive inputs, \libcubwt runs at $1.4$ to $1.5$\,ns per character and is $2.0$ to $3.2$ times faster than \code{cub6} ($13.7$ times on \texttt{pitches}).
While \code{fast} performs better, it is still a factor of $1.8$ to $3.1$ ($5.8$ on \texttt{pitches}) slower than \libcubwt.
Additionally, it needs more space than \libcubwt.
However, on the repetitive inputs, \libcubwt slows down to $2.2$ to $10.0$\,ns per character.
Here, \code{cub6} is faster than \libcubwt on five of the nine inputs (by a factor of up to $5$ on \texttt{world\_leaders}), on par on three (ratios $0.98$ to $1.04$), and slower only on \texttt{Escherichia\_Coli} ($1.76$).
Interestingly, the latter is the repetitive input with the largest $s$ ($0.10n$) and the smallest $n / r$ ratio.
Here, also \code{space} is faster than \libcubwt on the same five inputs, at $5.9$ to $12.6$ bytes per character on all nine.
Overall, \libcubwt thus remains the better choice for non-repetitive texts that fit into its memory.
Our construction is preferable for repetitive and for larger texts.

\begin{table}[t]
\caption{Time relative to \libcubwt ($t/t_{\mathrm{cu}}$, lower is better) and
peak memory in bytes per input character (B/n, lower is better) at the
largest size all five handle.
Here $t_{\mathrm{cu}}$ is the time of \libcubwt, which needs $20.4$ bytes per character throughout.}\label{tab:libcubwt}
\centering\small
\setlength{\tabcolsep}{4pt}
\begin{tabular}{@{}lrrrrrrrrrr@{}}
\toprule
& & & \multicolumn{2}{c}{\code{fast}} & \multicolumn{2}{c}{\code{cub6}}
& \multicolumn{2}{c}{\code{space}} & \multicolumn{2}{c}{\code{dbl}} \\
\cmidrule(lr){4-5}\cmidrule(lr){6-7}\cmidrule(lr){8-9}\cmidrule(l){10-11}
text & MiB & $t_{\mathrm{cu}}$ [s] & $t/t_{\mathrm{cu}}$ & B/n
& $t/t_{\mathrm{cu}}$ & B/n & $t/t_{\mathrm{cu}}$ & B/n & $t/t_{\mathrm{cu}}$ & B/n \\
\midrule
\texttt{dblp.xml} & 282.4 & 0.45 & 1.84 & 25.7 & 1.99 & 7.5 & 2.25 & 7.5 & 3.29 & 14.1 \\
\texttt{dna} & 385.2 & 0.55 & 3.10 & 45.8 & 3.22 & 13.5 & 3.56 & 11.6 & 3.65 & 22.9 \\
\texttt{english} & 381.5 & 0.59 & 2.76 & 48.8 & 2.74 & 11.9 & 2.99 & 11.7 & 3.28 & 20.4 \\
\texttt{pitches} & 53.2 & 0.08 & 5.83 & 55.2 & 13.73 & 17.0 & 14.68 & 14.6 & 4.37 & 22.3 \\
\texttt{proteins} & 596.0 & 0.91 & 2.68 & 35.7 & 2.97 & 15.4 & 3.31 & 11.1 & 3.16 & 17.2 \\
\texttt{sources} & 201.1 & 0.30 & 2.39 & 32.9 & 2.47 & 11.6 & 2.76 & 9.9 & 3.17 & 16.7 \\
\midrule
\texttt{cere} & 439.9 & 3.12 & 0.49 & 16.8 & 0.49 & 16.8 & 0.52 & 12.6 & 0.52 & 12.6 \\
\texttt{coreutils} & 195.8 & 0.68 & 0.90 & 8.3 & 1.04 & 8.3 & 1.12 & 6.3 & 0.80 & 6.3 \\
\texttt{einstein.de.txt} & 88.5 & 0.37 & 0.65 & 9.7 & 0.64 & 9.7 & 0.75 & 7.3 & 0.76 & 7.3 \\
\texttt{einstein.en.txt} & 446.0 & 1.83 & 0.98 & 13.3 & 0.98 & 13.3 & 1.03 & 9.7 & 0.96 & 9.7 \\
\texttt{Escherichia\_Coli} & 107.5 & 0.25 & 1.76 & 11.9 & 1.76 & 11.7 & 1.89 & 8.6 & 1.71 & 9.6 \\
\texttt{influenza} & 147.6 & 0.39 & 1.04 & 9.8 & 1.04 & 9.8 & 1.16 & 7.3 & 1.19 & 7.3 \\
\texttt{kernel} & 246.0 & 0.86 & 0.77 & 9.7 & 0.77 & 9.7 & 0.82 & 7.3 & 0.84 & 7.3 \\
\texttt{para} & 409.4 & 1.80 & 0.70 & 13.8 & 0.69 & 13.8 & 0.74 & 10.1 & 0.74 & 10.1 \\
\texttt{world\_leaders} & 44.8 & 0.47 & 0.20 & 6.1 & 0.20 & 6.1 & 0.27 & 5.9 & 0.28 & 6.8 \\
\bottomrule
\end{tabular}
\end{table}

\begin{figure}[tbp]
\centering
\includegraphics[width=0.95\textwidth]{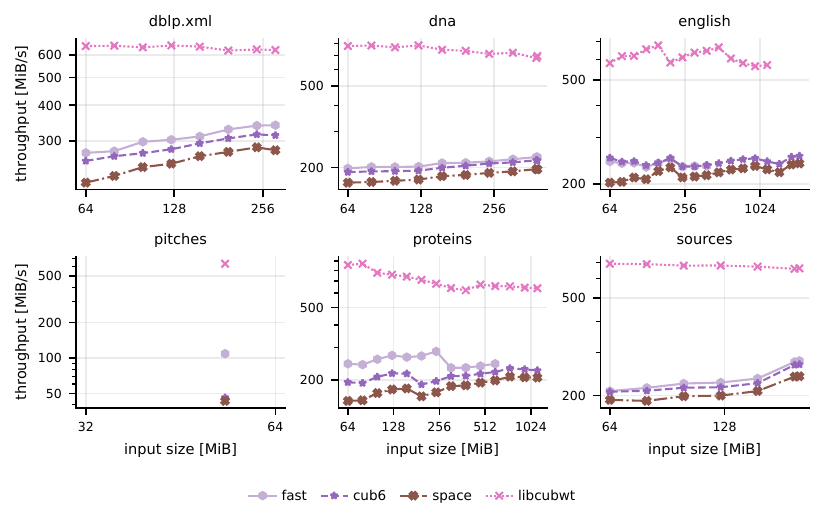}\\[1ex]
\includegraphics[width=0.95\textwidth]{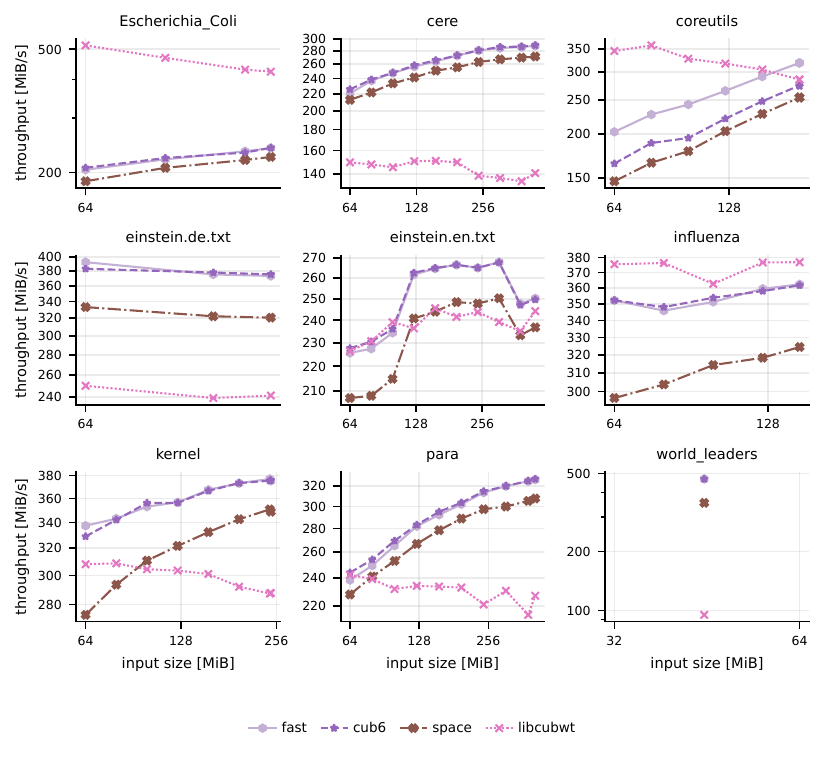}
\caption{Throughput over input size of \code{fast}, \code{cub6}, \code{space}, \code{dbl}, configurations and \libcubwt.}\label{fig:throughput}
\end{figure}

\begin{figure}[tbp]
\centering
\includegraphics[width=0.97\textwidth]{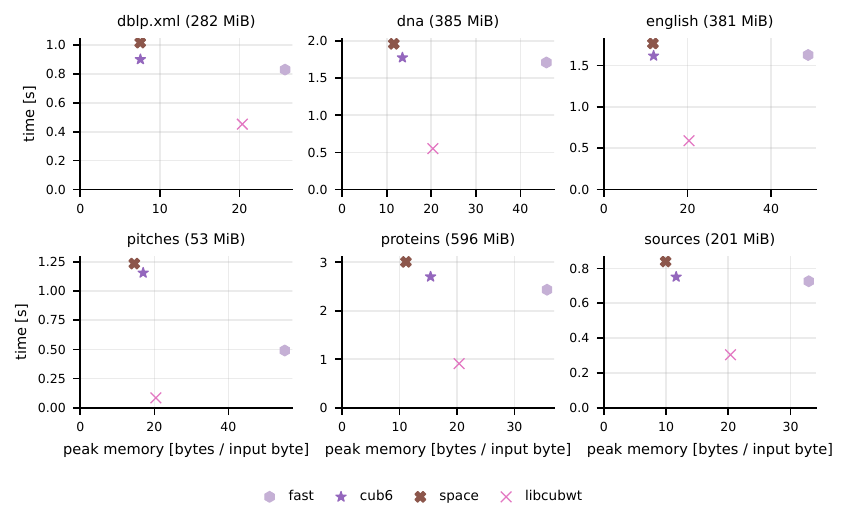}\\[1ex]
\includegraphics[width=0.97\textwidth]{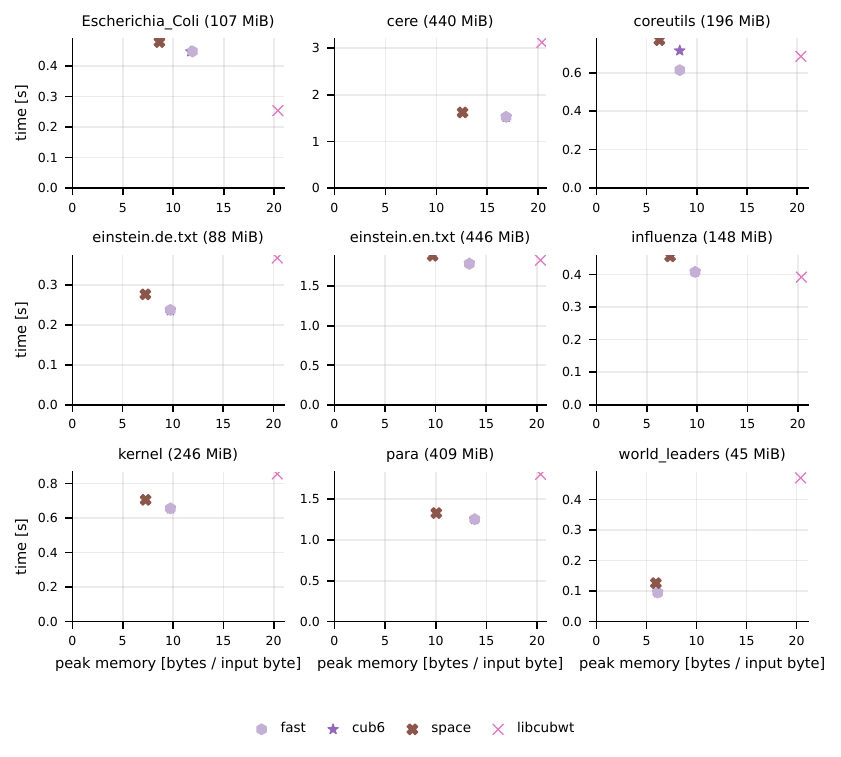}
\caption{Peak memory per input character and time of \code{fast}, \code{cub6}, \code{space}, \code{dbl},
configurations and \libcubwt at the largest size that all of them handle.}\label{fig:spacetime-selected}
\end{figure}

\subparagraph*{Where the time goes.}
On the non-repetitive inputs, the suffix sort is the most expensive phase of \code{cub6}, with $39$ to $88\%$ of the time.
On \texttt{english} and \texttt{dna}, its rounds are work-bound, e.g., $571$ rounds of $8.4$\,ms on average on the whole \texttt{english}.
On \texttt{pitches}, it needs $5\,485$ rounds of $0.19$\,ms on average, dominated by the fixed part.
Hence, \code{fast}, with a single chunk and $737$ rounds, and \code{dbl}, with $10$ rounds ($5$ to $10$ on all inputs), are $2.4$ and $3.1$ times faster than \code{cub6} there.
On \texttt{english}, however, \code{dbl} is $1.2$ times slower than \code{cub6} at every input size.
On the repetitive inputs, the suffix sort takes only $2$ to $42\%$ of the time.
Still, the dictionary phases take $53$ to $81\%$.
The parse-rank phase, i.e., the prefix doubling of $\SA(\Parse)$, alone takes $38$ to $45\%$ on \texttt{cere}, \texttt{einstein.de.txt}, \texttt{einstein.en.txt}, \texttt{kernel}, and \texttt{para}.
This is due to all $d = 0.03n$ to $0.15n$ phrase occurrences that have to be kept instead of the far fewer distinct phrases.

\section{Conclusion}\label{sec:conclusion}
We gave a variant of the prefix-free-parsing-based Burrows--Wheeler transform construction algorithm that works well on the GPU.
Additionally, we showed how to make this simple algorithm more memory efficient and engineer its performance using different sorting routine.
Overall, we reduce the peak on \texttt{english} by a factor of four, so that a $24$\,GB GPU processes a $1.8$ times larger prefix than with \libcubwt.
Unfortunately, this only works well for repetitive inputs.
Here, our construction is faster than or on par with \libcubwt on eight of nine inputs.
On non-repetitive inputs, we are significantly slower, because the suffix sort is the bottleneck.
On repetitive texts, the dictionary phases dominate, because they process all $d$ phrase occurrences.
Working on the $|\Dset|$ distinct phrases plus a compact occurrence structure would improve this on these texts.

\bibliography{lit.dblp}

\begin{thebibliography}{100}

\bibitem{AdlerBHS25}
Enno Adler, Stefan B{\"{o}}ttcher, Rita Hartel, and Cederic~Alexander
  Steininger.
\newblock {IBB}: Fast burrows-wheeler transform construction for length-diverse
  {DNA} data.
\newblock In {\em {SEA}}, volume 338 of {\em LIPIcs}, pages 2:1--2:18. Schloss
  Dagstuhl - Leibniz-Zentrum f{\"{u}}r Informatik, 2025.
\newblock \href {https://doi.org/10.4230/LIPIcs.SEA.2025.2}
  {\path{doi:10.4230/LIPIcs.SEA.2025.2}}.

\bibitem{SACABench19}
Johannes Bahne, Nico Bertram, Marvin B{\"{o}}cker, Jonas Bode, Johannes
  Fischer, Hermann Foot, Florian Grieskamp, Florian Kurpicz, Marvin
  L{\"{o}}bel, Oliver Magiera, Rosa Pink, David Piper, and Christopher Poeplau.
\newblock {SACABench}: Benchmarking suffix array construction.
\newblock In {\em {SPIRE}}, volume 11811 of {\em Lecture Notes in Computer
  Science}, pages 407--416. Springer, 2019.
\newblock \href {https://doi.org/10.1007/978-3-030-32686-9\_29}
  {\path{doi:10.1007/978-3-030-32686-9\_29}}.

\bibitem{Baier16}
Uwe Baier.
\newblock Linear-time suffix sorting - a new approach for suffix array
  construction.
\newblock In {\em {CPM}}, volume~54 of {\em LIPIcs}, pages 23:1--23:12. Schloss
  Dagstuhl - Leibniz-Zentrum f{\"{u}}r Informatik, 2016.
\newblock \href {https://doi.org/10.4230/LIPIcs.CPM.2016.23}
  {\path{doi:10.4230/LIPIcs.CPM.2016.23}}.

\bibitem{BannaiKKP24}
Hideo Bannai, Juha K{\"{a}}rkk{\"{a}}inen, Dominik K{\"{o}}ppl, and Marcin
  Piatkowski.
\newblock Constructing and indexing the bijective and extended burrows-wheeler
  transform.
\newblock {\em Inf. Comput.}, 297:105153, 2024.
\newblock \href {https://doi.org/10.1016/j.ic.2024.105153}
  {\path{doi:10.1016/j.ic.2024.105153}}.

\bibitem{BauerCR13}
Markus~J. Bauer, Anthony~J. Cox, and Giovanna Rosone.
\newblock Lightweight algorithms for constructing and inverting the {BWT} of
  string collections.
\newblock {\em Theor. Comput. Sci.}, 483:134--148, 2013.
\newblock \href {https://doi.org/10.1016/j.tcs.2012.02.002}
  {\path{doi:10.1016/j.tcs.2012.02.002}}.

\bibitem{Belazzougui14}
Djamal Belazzougui.
\newblock Linear time construction of compressed text indices in compact space.
\newblock In {\em {STOC}}, pages 148--193. ACM, 2014.
\newblock \href {https://doi.org/10.1145/2591796.2591885}
  {\path{doi:10.1145/2591796.2591885}}.

\bibitem{BellH12}
Nathan Bell and Jared Hoberock.
\newblock Thrust: A productivity-oriented library for {CUDA}.
\newblock In Wen mei W.~Hwu, editor, {\em {GPU} Computing Gems Jade Edition},
  pages 359--371. Morgan Kaufmann, 2012.
\newblock \href {https://doi.org/10.1016/B978-0-12-385963-1.00026-5}
  {\path{doi:10.1016/B978-0-12-385963-1.00026-5}}.

\bibitem{BellerZGO13}
Timo Beller, Maike Zwerger, Simon Gog, and Enno Ohlebusch.
\newblock Space-efficient construction of the burrows-wheeler transform.
\newblock In {\em {SPIRE}}, volume 8214 of {\em Lecture Notes in Computer
  Science}, pages 5--16. Springer, 2013.
\newblock \href {https://doi.org/10.1007/978-3-319-02432-5\_5}
  {\path{doi:10.1007/978-3-319-02432-5\_5}}.

\bibitem{BertramEF21}
Nico Bertram, Jonas Ellert, and Johannes Fischer.
\newblock Lyndon words accelerate suffix sorting.
\newblock In {\em {ESA}}, volume 204 of {\em LIPIcs}, pages 15:1--15:13.
  Schloss Dagstuhl - Leibniz-Zentrum f{\"{u}}r Informatik, 2021.
\newblock \href {https://doi.org/10.4230/LIPIcs.ESA.2021.15}
  {\path{doi:10.4230/LIPIcs.ESA.2021.15}}.

\bibitem{BezKLS23}
Dominik Bez, Florian Kurpicz, Hans{-}Peter Lehmann, and Peter Sanders.
\newblock High performance construction of {RecSplit} based minimal perfect
  hash functions.
\newblock In {\em {ESA}}, volume 274 of {\em LIPIcs}, pages 19:1--19:16.
  Schloss Dagstuhl - Leibniz-Zentrum f{\"{u}}r Informatik, 2023.
\newblock \href {https://doi.org/10.4230/LIPIcs.ESA.2023.19}
  {\path{doi:10.4230/LIPIcs.ESA.2023.19}}.

\bibitem{BingmannDFKOS22}
Timo Bingmann, Patrick Dinklage, Johannes Fischer, Florian Kurpicz, Enno
  Ohlebusch, and Peter Sanders.
\newblock Scalable text index construction.
\newblock In {\em {Algorithms for Big Data}}, volume 13201 of {\em Lecture
  Notes in Computer Science}, pages 252--284. Springer, 2022.
\newblock \href {https://doi.org/10.1007/978-3-031-21534-6\_14}
  {\path{doi:10.1007/978-3-031-21534-6\_14}}.

\bibitem{BingmannFO16}
Timo Bingmann, Johannes Fischer, and Vitaly Osipov.
\newblock Inducing suffix and lcp arrays in external memory.
\newblock In {\em {ALENEX}}, pages 88--102. SIAM, 2013.
\newblock \href {https://doi.org/10.1137/1.9781611972931.8}
  {\path{doi:10.1137/1.9781611972931.8}}.

\bibitem{BingmannGK18}
Timo Bingmann, Simon Gog, and Florian Kurpicz.
\newblock Scalable construction of text indexes with thrill.
\newblock In {\em {IEEE BigData}}, pages 634--643. IEEE, 2018.
\newblock \href {https://doi.org/10.1109/BigData.2018.8622171}
  {\path{doi:10.1109/BigData.2018.8622171}}.

\bibitem{BonizzoniVPPR21}
Paola Bonizzoni, Gianluca~Della Vedova, Yuri Pirola, Marco Previtali, and
  Raffaella Rizzi.
\newblock Computing the multi-string {BWT} and {LCP} array in external memory.
\newblock {\em Theor. Comput. Sci.}, 862:42--58, 2021.
\newblock \href {https://doi.org/10.1016/j.tcs.2020.11.041}
  {\path{doi:10.1016/j.tcs.2020.11.041}}.

\bibitem{BonomoMRRS14}
Silvia Bonomo, Sabrina Mantaci, Antonio Restivo, Giovanna Rosone, and Marinella
  Sciortino.
\newblock Sorting conjugates and suffixes of words in a multiset.
\newblock {\em Int. J. Found. Comput. Sci.}, 25(8):1161--, 2014.
\newblock \href {https://doi.org/10.1142/S0129054114400309}
  {\path{doi:10.1142/S0129054114400309}}.

\bibitem{BoucherGKLMM19}
Christina Boucher, Travis Gagie, Alan Kuhnle, Ben Langmead, Giovanni Manzini,
  and Taher Mun.
\newblock Prefix-free parsing for building big {BWTs.}
\newblock {\em Algorithms Mol. Biol.}, 14(1):13:1--13:15, 2019.
\newblock \href {https://doi.org/10.1186/s13015-019-0148-5}
  {\path{doi:10.1186/s13015-019-0148-5}}.

\bibitem{BurenJKHS19}
Florian B{\"{u}}ren, Daniel J{\"{u}}nger, Robin Kobus, Christian Hundt, and
  Bertil Schmidt.
\newblock Suffix array construction on multi-{GPU} systems.
\newblock In {\em {HPDC}}, pages 183--194. ACM, 2019.
\newblock \href {https://doi.org/10.1145/3307681.3325961}
  {\path{doi:10.1145/3307681.3325961}}.

\bibitem{BurrowsW94}
Michael Burrows and David~J. Wheeler.
\newblock A block-sorting lossless data compression algorithm.
\newblock Technical report, 1994.

\bibitem{CoxGRS16}
Anthony~J. Cox, Fabio Garofalo, Giovanna Rosone, and Marinella Sciortino.
\newblock Lightweight {LCP} construction for very large collections of strings.
\newblock {\em J. Discrete Algorithms}, 37:17--33, 2016.
\newblock \href {https://doi.org/10.1016/j.jda.2016.03.003}
  {\path{doi:10.1016/j.jda.2016.03.003}}.

\bibitem{CrauserF02}
Andreas Crauser and Paolo Ferragina.
\newblock A theoretical and experimental study on the construction of suffix
  arrays in external memory.
\newblock {\em Algorithmica}, 32(1):1--35, 2002.
\newblock \href {https://doi.org/10.1007/s00453-001-0051-5}
  {\path{doi:10.1007/s00453-001-0051-5}}.

\bibitem{CrochemoreGKL15}
Maxime Crochemore, Roberto Grossi, Juha K{\"{a}}rkk{\"{a}}inen, and Gad~M.
  Landau.
\newblock Computing the burrows-wheeler transform in place and in small space.
\newblock {\em J. Discrete Algorithms}, 32:44--52, 2015.
\newblock \href {https://doi.org/10.1016/j.jda.2015.01.004}
  {\path{doi:10.1016/j.jda.2015.01.004}}.

\bibitem{DementievKMS08}
Roman Dementiev, Juha K{\"{a}}rkk{\"{a}}inen, Jens Mehnert, and Peter Sanders.
\newblock Better external memory suffix array construction.
\newblock {\em ACM J. Exp. Algorithmics}, 12:3.4:1--3.4:24, 2008.
\newblock \href {https://doi.org/10.1145/1227161.1402296}
  {\path{doi:10.1145/1227161.1402296}}.

\bibitem{DeoK13}
Mrinal Deo and Sean Keely.
\newblock Parallel suffix array and least common prefix for the {GPU.}
\newblock In {\em {PPoPP}}, pages 197--206. ACM, 2013.
\newblock \href {https://doi.org/10.1145/2442516.2442536}
  {\path{doi:10.1145/2442516.2442536}}.

\bibitem{DeshpandeN15}
Aditya Deshpande and P.~J. Narayanan.
\newblock Fast burrows wheeler compression using all-cores.
\newblock In {\em {IPDPS Workshops}}, pages 628--636. IEEE Computer Society,
  2015.
\newblock \href {https://doi.org/10.1109/IPDPSW.2015.53}
  {\path{doi:10.1109/IPDPSW.2015.53}}.

\bibitem{Diaz-DominguezN23}
Diego D{\'{\i}}az{-}Dom{\'{\i}}nguez and Gonzalo Navarro.
\newblock Efficient construction of the {BWT} for repetitive text using string
  compression.
\newblock {\em Inf. Comput.}, 294:105088, 2023.
\newblock \href {https://doi.org/10.1016/j.ic.2023.105088}
  {\path{doi:10.1016/j.ic.2023.105088}}.

\bibitem{EdwardsV14}
James~Alexander Edwards and Uzi Vishkin.
\newblock Parallel algorithms for burrows-wheeler compression and
  decompression.
\newblock {\em Theor. Comput. Sci.}, 525:10--22, 2014.
\newblock \href {https://doi.org/10.1016/j.tcs.2013.10.009}
  {\path{doi:10.1016/j.tcs.2013.10.009}}.

\bibitem{EgidiLMT19}
Lavinia Egidi, Felipe~A. Louza, Giovanni Manzini, and Guilherme~P. Telles.
\newblock External memory {BWT} and {LCP} computation for sequence collections
  with applications.
\newblock {\em Algorithms Mol. Biol.}, 14(1):6:1--6:15, 2019.
\newblock \href {https://doi.org/10.1186/s13015-019-0140-0}
  {\path{doi:10.1186/s13015-019-0140-0}}.

\bibitem{EgidiM20}
Lavinia Egidi and Giovanni Manzini.
\newblock Lightweight merging of compressed indices based on {BWT} variants.
\newblock {\em Theor. Comput. Sci.}, 812:214--229, 2020.
\newblock \href {https://doi.org/10.1016/j.tcs.2019.11.001}
  {\path{doi:10.1016/j.tcs.2019.11.001}}.

\bibitem{FerraginaGM12}
Paolo Ferragina, Travis Gagie, and Giovanni Manzini.
\newblock Lightweight data indexing and compression in external memory.
\newblock {\em Algorithmica}, 63(3):707--730, 2012.
\newblock \href {https://doi.org/10.1007/s00453-011-9535-0}
  {\path{doi:10.1007/s00453-011-9535-0}}.

\bibitem{FerraginaM00}
Paolo Ferragina and Giovanni Manzini.
\newblock Opportunistic data structures with applications.
\newblock In {\em {FOCS}}, pages 390--398. IEEE Computer Society, 2000.
\newblock \href {https://doi.org/10.1109/SFCS.2000.892127}
  {\path{doi:10.1109/SFCS.2000.892127}}.

\bibitem{PizzaChili}
Paolo Ferragina and Gonzalo Navarro.
\newblock Pizza\&chili corpus.
\newblock \url{http://pizzachili.dcc.uchile.cl/}, 2005.

\bibitem{FerroPGLB26}
Eddie Ferro, Tyler Pencinger, Oded Green, Mahsa Lotfollahi, and Christina
  Boucher.
\newblock {Rapid-PFP}: Accelerating prefix-free parsing with {GPU} parallelism.
\newblock bioRxiv preprint, 2026.
\newblock \href {https://doi.org/10.64898/2026.04.28.721411}
  {\path{doi:10.64898/2026.04.28.721411}}.

\bibitem{FischerK17}
Johannes Fischer and Florian Kurpicz.
\newblock Dismantling {DivSufSort.}
\newblock In {\em {PSC}}, pages 62--76. Department of Theoretical Computer
  Science, Faculty of Information Technology, Czech Technical University in
  Prague, 2017.

\bibitem{FischerK19}
Johannes Fischer and Florian Kurpicz.
\newblock Lightweight distributed suffix array construction.
\newblock In {\em {ALENEX}}, pages 27--38. SIAM, 2019.
\newblock \href {https://doi.org/10.1137/1.9781611975499.3}
  {\path{doi:10.1137/1.9781611975499.3}}.

\bibitem{FlickA15}
Patrick Flick and Srinivas Aluru.
\newblock Parallel distributed memory construction of suffix and longest common
  prefix arrays.
\newblock In {\em {SC}}, pages 16:1--16:10. ACM, 2015.
\newblock \href {https://doi.org/10.1145/2807591.2807609}
  {\path{doi:10.1145/2807591.2807609}}.

\bibitem{FranzrebBR25}
Marco Franzreb, Martin Burtscher, and Stephan Rudolph.
\newblock {GPU} implementation of the wavelet tree.
\newblock {\em CoRR}, abs/2505.03372, 2025.
\newblock \href {https://doi.org/10.48550/arXiv.2505.03372}
  {\path{doi:10.48550/arXiv.2505.03372}}.

\bibitem{Fuentes-SepulvedaNN19}
Jos{\'{e}} Fuentes{-}Sep{\'{u}}lveda, Gonzalo Navarro, and Yakov Nekrich.
\newblock Space-efficient computation of the burrows-wheeler transform.
\newblock In {\em {DCC}}, pages 132--141. IEEE, 2019.
\newblock \href {https://doi.org/10.1109/DCC.2019.00021}
  {\path{doi:10.1109/DCC.2019.00021}}.

\bibitem{Fuentes-SepulvedaNN20}
Jos{\'{e}} Fuentes{-}Sep{\'{u}}lveda, Gonzalo Navarro, and Yakov Nekrich.
\newblock Parallel computation of the burrows wheeler transform in compact
  space.
\newblock {\em Theor. Comput. Sci.}, 812:123--136, 2020.
\newblock \href {https://doi.org/10.1016/j.tcs.2019.09.030}
  {\path{doi:10.1016/j.tcs.2019.09.030}}.

\bibitem{GagieNP20}
Travis Gagie, Gonzalo Navarro, and Nicola Prezza.
\newblock Fully functional suffix trees and optimal text searching in
  {BWT}-runs bounded space.
\newblock {\em J. ACM}, 67(1):2:1--2:54, 2020.
\newblock \href {https://doi.org/10.1145/3375890} {\path{doi:10.1145/3375890}}.

\bibitem{GalluzzoGRR26}
Ylenia Galluzzo, Raffaele Giancarlo, Mario Randazzo, and Simona~E. Rombo.
\newblock Burrows wheeler transform on a large scale: Algorithms implemented in
  apache spark.
\newblock {\em Data}, 11(3), 2026.
\newblock \href {https://doi.org/10.3390/data11030048}
  {\path{doi:10.3390/data11030048}}.

\bibitem{libsais}
Ilya Grebnov.
\newblock libsais: suffix array and {BWT} construction based on induced
  sorting.
\newblock \url{https://github.com/IlyaGrebnov/libsais}, 2021.
\newblock Version 2.10.4.

\bibitem{libcubwt}
Ilya Grebnov.
\newblock libcubwt: {GPU}-accelerated suffix array and {BWT} construction.
\newblock \url{https://github.com/IlyaGrebnov/libcubwt}, 2023.
\newblock Version 1.6.3.

\bibitem{HaagKSS25}
Manuel Haag, Florian Kurpicz, Peter Sanders, and Matthias Schimek.
\newblock Fast and lightweight distributed suffix array construction.
\newblock In {\em {ESA}}, volume 351 of {\em LIPIcs}, pages 47:1--47:18.
  Schloss Dagstuhl - Leibniz-Zentrum f{\"{u}}r Informatik, 2025.
\newblock \href {https://doi.org/10.4230/LIPIcs.ESA.2025.47}
  {\path{doi:10.4230/LIPIcs.ESA.2025.47}}.

\bibitem{HayashiT13}
Shinya Hayashi and Kenjiro Taura.
\newblock Parallel and memory-efficient burrows-wheeler transform.
\newblock In {\em {IEEE BigData}}, pages 43--50. IEEE Computer Society, 2013.
\newblock \href {https://doi.org/10.1109/BigData.2013.6691757}
  {\path{doi:10.1109/BigData.2013.6691757}}.

\bibitem{HoltM14}
James Holt and Leonard McMillan.
\newblock Merging of multi-string {BWTs} with applications.
\newblock {\em Bioinform.}, 30(24):3524--3531, 2014.
\newblock \href {https://doi.org/10.1093/bioinformatics/btu584}
  {\path{doi:10.1093/bioinformatics/btu584}}.

\bibitem{HonLSSY07}
Wing{-}Kai Hon, Tak~Wah Lam, Kunihiko Sadakane, Wing{-}Kin Sung, and Siu{-}Ming
  Yiu.
\newblock A space and time efficient algorithm for constructing compressed
  suffix arrays.
\newblock {\em Algorithmica}, 48(1):23--36, 2007.
\newblock \href {https://doi.org/10.1007/s00453-006-1228-8}
  {\path{doi:10.1007/s00453-006-1228-8}}.

\bibitem{HonSS09}
Wing{-}Kai Hon, Kunihiko Sadakane, and Wing{-}Kin Sung.
\newblock Breaking a time-and-space barrier in constructing full-text indices.
\newblock {\em SIAM J. Comput.}, 38(6):2162--2178, 2009.
\newblock \href {https://doi.org/10.1137/070685373}
  {\path{doi:10.1137/070685373}}.

\bibitem{Karkkainen07}
Juha K{\"{a}}rkk{\"{a}}inen.
\newblock Fast {BWT} in small space by blockwise suffix sorting.
\newblock {\em Theor. Comput. Sci.}, 387(3):249--257, 2007.
\newblock \href {https://doi.org/10.1016/j.tcs.2007.07.018}
  {\path{doi:10.1016/j.tcs.2007.07.018}}.

\bibitem{KarkkainenKP15}
Juha K{\"{a}}rkk{\"{a}}inen, Dominik Kempa, and Simon~J. Puglisi.
\newblock Parallel external memory suffix sorting.
\newblock In {\em {CPM}}, volume 9133 of {\em Lecture Notes in Computer
  Science}, pages 329--342. Springer, 2015.
\newblock \href {https://doi.org/10.1007/978-3-319-19929-0\_28}
  {\path{doi:10.1007/978-3-319-19929-0\_28}}.

\bibitem{KarkkainenKPZ17}
Juha K{\"{a}}rkk{\"{a}}inen, Dominik Kempa, Simon~J. Puglisi, and Bella
  Zhukova.
\newblock Engineering external memory induced suffix sorting.
\newblock In {\em {ALENEX}}, pages 98--108. SIAM, 2017.
\newblock \href {https://doi.org/10.1137/1.9781611974768.8}
  {\path{doi:10.1137/1.9781611974768.8}}.

\bibitem{KarkkainenSB06}
Juha K{\"{a}}rkk{\"{a}}inen, Peter Sanders, and Stefan Burkhardt.
\newblock Linear work suffix array construction.
\newblock {\em J. ACM}, 53(6):918--936, 2006.
\newblock \href {https://doi.org/10.1145/1217856.1217858}
  {\path{doi:10.1145/1217856.1217858}}.

\bibitem{KarpR87}
Richard~M. Karp and Michael~O. Rabin.
\newblock Efficient randomized pattern-matching algorithms.
\newblock {\em IBM J. Res. Dev.}, 31(2):249--260, 1987.
\newblock \href {https://doi.org/10.1147/rd.312.0249}
  {\path{doi:10.1147/rd.312.0249}}.

\bibitem{Kempa19}
Dominik Kempa.
\newblock Optimal construction of compressed indexes for highly repetitive
  texts.
\newblock In {\em {SODA}}, pages 1344--1357. SIAM, 2019.
\newblock \href {https://doi.org/10.1137/1.9781611975482.82}
  {\path{doi:10.1137/1.9781611975482.82}}.

\bibitem{KempaK19}
Dominik Kempa and Tomasz Kociumaka.
\newblock String synchronizing sets: sublinear-time {BWT} construction and
  optimal {LCE} data structure.
\newblock In {\em {STOC}}, pages 756--767. ACM, 2019.
\newblock \href {https://doi.org/10.1145/3313276.3316368}
  {\path{doi:10.1145/3313276.3316368}}.

\bibitem{KempaK20}
Dominik Kempa and Tomasz Kociumaka.
\newblock Resolution of the burrows-wheeler transform conjecture.
\newblock In {\em {FOCS}}, pages 1002--1013. IEEE, 2020.
\newblock \href {https://doi.org/10.1109/FOCS46700.2020.00097}
  {\path{doi:10.1109/FOCS46700.2020.00097}}.

\bibitem{KhanRMDP24}
Jamshed Khan, Tobias Rubel, Erin~K. Molloy, Laxman Dhulipala, and Rob Patro.
\newblock Fast, parallel, and cache-friendly suffix array construction.
\newblock {\em Algorithms Mol. Biol.}, 19(1):16, 2024.
\newblock \href {https://doi.org/10.1186/s13015-024-00263-5}
  {\path{doi:10.1186/s13015-024-00263-5}}.

\bibitem{KimSPP05}
Dong~Kyue Kim, Jeong~Seop Sim, Heejin Park, and Kunsoo Park.
\newblock Constructing suffix arrays in linear time.
\newblock {\em J. Discrete Algorithms}, 3(2-4):126--142, 2005.
\newblock \href {https://doi.org/10.1016/j.jda.2004.08.019}
  {\path{doi:10.1016/j.jda.2004.08.019}}.

\bibitem{KitajimaN99}
Joao~Paulo Kitajima and Gonzalo Navarro.
\newblock A fast distributed suffix array generation algorithm.
\newblock In {\em {SPIRE/CRIWG}}, pages 97--105. IEEE Computer Society, 1999.
\newblock \href {https://doi.org/10.1109/SPIRE.1999.796583}
  {\path{doi:10.1109/SPIRE.1999.796583}}.

\bibitem{KoA05}
Pang Ko and Srinivas Aluru.
\newblock Space efficient linear time construction of suffix arrays.
\newblock {\em J. Discrete Algorithms}, 3(2-4):143--156, 2005.
\newblock \href {https://doi.org/10.1016/j.jda.2004.08.002}
  {\path{doi:10.1016/j.jda.2004.08.002}}.

\bibitem{KopplHHS20}
Dominik K{\"{o}}ppl, Daiki Hashimoto, Diptarama Hendrian, and Ayumi Shinohara.
\newblock In-place bijective burrows-wheeler transforms.
\newblock In {\em {CPM}}, volume 161 of {\em LIPIcs}, pages 21:1--21:15.
  Schloss Dagstuhl - Leibniz-Zentrum f{\"{u}}r Informatik, 2020.
\newblock \href {https://doi.org/10.4230/LIPIcs.CPM.2020.21}
  {\path{doi:10.4230/LIPIcs.CPM.2020.21}}.

\bibitem{KuhnleMBGLM20}
Alan Kuhnle, Taher Mun, Christina Boucher, Travis Gagie, Ben Langmead, and
  Giovanni Manzini.
\newblock Efficient construction of a complete index for pan-genomics read
  alignment.
\newblock {\em J. Comput. Biol.}, 27(4):500--513, 2020.
\newblock \href {https://doi.org/10.1089/cmb.2019.0309}
  {\path{doi:10.1089/cmb.2019.0309}}.

\bibitem{KullaS07}
Fabian Kulla and Peter Sanders.
\newblock Scalable parallel suffix array construction.
\newblock {\em Parallel Comput.}, 33(9):605--612, 2007.
\newblock \href {https://doi.org/10.1016/j.parco.2007.06.004}
  {\path{doi:10.1016/j.parco.2007.06.004}}.

\bibitem{LabeitSB17}
Julian Labeit, Julian Shun, and Guy~E. Blelloch.
\newblock Parallel lightweight wavelet tree, suffix array and {FM}-index
  construction.
\newblock {\em J. Discrete Algorithms}, 43:2--17, 2017.
\newblock \href {https://doi.org/10.1016/j.jda.2017.04.001}
  {\path{doi:10.1016/j.jda.2017.04.001}}.

\bibitem{LangmeadTPS09}
Ben Langmead, Cole Trapnell, Mihai Pop, and Steven~L. Salzberg.
\newblock Ultrafast and memory-efficient alignment of short {DNA} sequences to
  the human genome.
\newblock {\em Genome Biology}, 10(3):R25, 2009.
\newblock \href {https://doi.org/10.1186/gb-2009-10-3-r25}
  {\path{doi:10.1186/gb-2009-10-3-r25}}.

\bibitem{LaoNCX18}
Bin Lao, Ge~Nong, Wai~Hong Chan, and Jing~Yi Xie.
\newblock Fast in-place suffix sorting on a multicore computer.
\newblock {\em IEEE Trans. Computers}, 67(12):1737--1749, 2018.
\newblock \href {https://doi.org/10.1109/TC.2018.2842050}
  {\path{doi:10.1109/TC.2018.2842050}}.

\bibitem{LarssonS07}
N.~Jesper Larsson and Kunihiko Sadakane.
\newblock Faster suffix sorting.
\newblock {\em Theor. Comput. Sci.}, 387(3):258--272, 2007.
\newblock \href {https://doi.org/10.1016/j.tcs.2007.07.017}
  {\path{doi:10.1016/j.tcs.2007.07.017}}.

\bibitem{Li24}
Heng Li.
\newblock {BWT} construction and search at the terabase scale.
\newblock {\em Bioinform.}, 40(12), 2024.
\newblock \href {https://doi.org/10.1093/bioinformatics/btae717}
  {\path{doi:10.1093/bioinformatics/btae717}}.

\bibitem{LiD09}
Heng Li and Richard Durbin.
\newblock Fast and accurate short read alignment with burrows-wheeler
  transform.
\newblock {\em Bioinform.}, 25(14):1754--1760, 2009.
\newblock \href {https://doi.org/10.1093/bioinformatics/btp324}
  {\path{doi:10.1093/bioinformatics/btp324}}.

\bibitem{LiYLLYKW09}
Ruiqiang Li, Chang Yu, Yingrui Li, Tak~Wah Lam, Siu{-}Ming Yiu, Karsten
  Kristiansen, and Jun Wang.
\newblock {SOAP2}: an improved ultrafast tool for short read alignment.
\newblock {\em Bioinform.}, 25(15):1966--1967, 2009.
\newblock \href {https://doi.org/10.1093/bioinformatics/btp336}
  {\path{doi:10.1093/bioinformatics/btp336}}.

\bibitem{LiLH22}
Zhize Li, Jian Li, and Hongwei Huo.
\newblock Optimal in-place suffix sorting.
\newblock {\em Inf. Comput.}, 285(Part):104818, 2022.
\newblock \href {https://doi.org/10.1016/j.ic.2021.104818}
  {\path{doi:10.1016/j.ic.2021.104818}}.

\bibitem{LiaoMZT15}
Gang Liao, Longfei Ma, Guangming Zang, and Lin Tang.
\newblock Parallel {DC3} algorithm for suffix array construction on many-core
  accelerators.
\newblock In {\em {CCGRID}}, pages 1155--1158. IEEE Computer Society, 2015.
\newblock \href {https://doi.org/10.1109/CCGrid.2015.56}
  {\path{doi:10.1109/CCGrid.2015.56}}.

\bibitem{LiuZW16}
Bo~Liu, Dixian Zhu, and Yadong Wang.
\newblock {deBWT}: parallel construction of burrows-wheeler transform for large
  collection of genomes with de bruijn-branch encoding.
\newblock {\em Bioinform.}, 32(12):174--182, 2016.
\newblock \href {https://doi.org/10.1093/bioinformatics/btw266}
  {\path{doi:10.1093/bioinformatics/btw266}}.

\bibitem{LiuLL14}
Chi{-}Man Liu, Ruibang Luo, and Tak~Wah Lam.
\newblock {GPU}-accelerated {BWT} construction for large collection of short
  reads.
\newblock {\em CoRR}, abs/1401.7457, 2014.
\newblock \href {https://doi.org/10.48550/arXiv.1401.7457}
  {\path{doi:10.48550/arXiv.1401.7457}}.

\bibitem{LiuHS16}
Yongchao Liu, Thomas Hankeln, and Bertil Schmidt.
\newblock Parallel and space-efficient construction of burrows-wheeler
  transform and suffix array for big genome data.
\newblock {\em IEEE ACM Trans. Comput. Biol. Bioinform.}, 13(3):592--598, 2016.
\newblock \href {https://doi.org/10.1109/TCBB.2015.2430314}
  {\path{doi:10.1109/TCBB.2015.2430314}}.

\bibitem{ManberM93}
Udi Manber and Eugene~W. Myers.
\newblock Suffix arrays: A new method for on-line string searches.
\newblock {\em SIAM J. Comput.}, 22(5):935--948, 1993.
\newblock \href {https://doi.org/10.1137/0222058} {\path{doi:10.1137/0222058}}.

\bibitem{MantaciRRS14}
Sabrina Mantaci, Antonio Restivo, Giovanna Rosone, and Marinella Sciortino.
\newblock Suffix array and lyndon factorization of a text.
\newblock {\em J. Discrete Algorithms}, 28:2--8, 2014.
\newblock \href {https://doi.org/10.1016/j.jda.2014.06.001}
  {\path{doi:10.1016/j.jda.2014.06.001}}.

\bibitem{Manzini01}
Giovanni Manzini.
\newblock An analysis of the burrows-wheeler transform.
\newblock {\em J. ACM}, 48(3):407--430, 2001.
\newblock \href {https://doi.org/10.1145/382780.382782}
  {\path{doi:10.1145/382780.382782}}.

\bibitem{ManziniF04}
Giovanni Manzini and Paolo Ferragina.
\newblock Engineering a lightweight suffix array construction algorithm.
\newblock {\em Algorithmica}, 40(1):33--50, 2004.
\newblock \href {https://doi.org/10.1007/s00453-004-1094-1}
  {\path{doi:10.1007/s00453-004-1094-1}}.

\bibitem{Masillo23}
Francesco Masillo.
\newblock Matching statistics speed up {BWT} construction.
\newblock In {\em {ESA}}, volume 274 of {\em LIPIcs}, pages 83:1--83:15.
  Schloss Dagstuhl - Leibniz-Zentrum f{\"{u}}r Informatik, 2023.
\newblock \href {https://doi.org/10.4230/LIPIcs.ESA.2023.83}
  {\path{doi:10.4230/LIPIcs.ESA.2023.83}}.

\bibitem{MenonBS11}
Rohith~K. Menon, Goutham~P. Bhat, and Michael~C. Schatz.
\newblock Rapid parallel genome indexing with {MapReduce}.
\newblock In {\em 2nd International Workshop on MapReduce and its Applications
  (MapReduce)}, pages 51--58, 2011.
\newblock \href {https://doi.org/10.1145/1996092.1996104}
  {\path{doi:10.1145/1996092.1996104}}.

\bibitem{CUB}
Duane Merrill.
\newblock {CUB}: Cooperative primitives for {CUDA} {C++}.
\newblock Part of the CUDA Core Compute Libraries,
  \url{https://github.com/NVIDIA/cccl}.
\newblock Version TODO.

\bibitem{divsufsort}
Yuta Mori.
\newblock libdivsufsort: A lightweight suffix-sorting library.
\newblock \url{https://github.com/y-256/libdivsufsort}, 2010.
\newblock Version 2.0.1.

\bibitem{MunroNN17}
J.~Ian Munro, Gonzalo Navarro, and Yakov Nekrich.
\newblock Space-efficient construction of compressed indexes in deterministic
  linear time.
\newblock In {\em {SODA}}, pages 408--424. SIAM, 2017.
\newblock \href {https://doi.org/10.1137/1.9781611974782.26}
  {\path{doi:10.1137/1.9781611974782.26}}.

\bibitem{NishimotoKT22}
Takaaki Nishimoto, Shunsuke Kanda, and Yasuo Tabei.
\newblock An optimal-time {RLBWT} construction in {BWT}-runs bounded space.
\newblock In {\em {ICALP}}, volume 229 of {\em LIPIcs}, pages 99:1--99:20.
  Schloss Dagstuhl - Leibniz-Zentrum f{\"{u}}r Informatik, 2022.
\newblock \href {https://doi.org/10.4230/LIPIcs.ICALP.2022.99}
  {\path{doi:10.4230/LIPIcs.ICALP.2022.99}}.

\bibitem{NishimotoT21}
Takaaki Nishimoto and Yasuo Tabei.
\newblock Optimal-time queries on {BWT}-runs compressed indexes.
\newblock In {\em {ICALP}}, volume 198 of {\em LIPIcs}, pages 101:1--101:15.
  Schloss Dagstuhl - Leibniz-Zentrum f{\"{u}}r Informatik, 2021.
\newblock \href {https://doi.org/10.4230/LIPIcs.ICALP.2021.101}
  {\path{doi:10.4230/LIPIcs.ICALP.2021.101}}.

\bibitem{Nong13}
Ge~Nong.
\newblock Practical linear-time o(1)-workspace suffix sorting for constant
  alphabets.
\newblock {\em ACM Trans. Inf. Syst.}, 31(3):15, 2013.
\newblock \href {https://doi.org/10.1145/2493175.2493180}
  {\path{doi:10.1145/2493175.2493180}}.

\bibitem{NongCHW15}
Ge~Nong, Wai~Hong Chan, Sheng~Qing Hu, and Yi~Wu.
\newblock Induced sorting suffixes in external memory.
\newblock {\em ACM Trans. Inf. Syst.}, 33(3):12:1--12:15, 2015.
\newblock \href {https://doi.org/10.1145/2699665} {\path{doi:10.1145/2699665}}.

\bibitem{NongZC11}
Ge~Nong, Sen Zhang, and Wai~Hong Chan.
\newblock Two efficient algorithms for linear time suffix array construction.
\newblock {\em IEEE Trans. Computers}, 60(10):1471--1484, 2011.
\newblock \href {https://doi.org/10.1109/TC.2010.188}
  {\path{doi:10.1109/TC.2010.188}}.

\bibitem{OhnoSTIS18}
Tatsuya Ohno, Kensuke Sakai, Yoshimasa Takabatake, Tomohiro I, and Hiroshi
  Sakamoto.
\newblock A faster implementation of online {RLBWT} and its application to
  {LZ77} parsing.
\newblock {\em J. Discrete Algorithms}, 52-53:18--28, 2018.
\newblock \href {https://doi.org/10.1016/j.jda.2018.11.002}
  {\path{doi:10.1016/j.jda.2018.11.002}}.

\bibitem{OkanoharaS09}
Daisuke Okanohara and Kunihiko Sadakane.
\newblock A linear-time burrows-wheeler transform using induced sorting.
\newblock In {\em {SPIRE}}, volume 5721 of {\em Lecture Notes in Computer
  Science}, pages 90--101. Springer, 2009.
\newblock \href {https://doi.org/10.1007/978-3-642-03784-9\_9}
  {\path{doi:10.1007/978-3-642-03784-9\_9}}.

\bibitem{Olbrich25}
Jannik Olbrich.
\newblock Fast and memory-efficient {BWT} construction of repetitive texts
  using lyndon grammars.
\newblock In {\em {ESA}}, volume 351 of {\em LIPIcs}, pages 60:1--60:19.
  Schloss Dagstuhl - Leibniz-Zentrum f{\"{u}}r Informatik, 2025.
\newblock \href {https://doi.org/10.4230/LIPIcs.ESA.2025.60}
  {\path{doi:10.4230/LIPIcs.ESA.2025.60}}.

\bibitem{OlbrichOB24}
Jannik Olbrich, Enno Ohlebusch, and Thomas B{\"{u}}chler.
\newblock Generic non-recursive suffix array construction.
\newblock {\em ACM Trans. Algorithms}, 20(2):18, 2024.
\newblock \href {https://doi.org/10.1145/3641854} {\path{doi:10.1145/3641854}}.

\bibitem{OlivaGB23}
Marco Oliva, Travis Gagie, and Christina Boucher.
\newblock Recursive prefix-free parsing for building big {BWTs.}
\newblock In {\em {DCC}}, pages 62--70. IEEE, 2023.
\newblock \href {https://doi.org/10.1109/DCC55655.2023.00014}
  {\path{doi:10.1109/DCC55655.2023.00014}}.

\bibitem{Osipov12}
Vitaly Osipov.
\newblock Parallel suffix array construction for shared memory architectures.
\newblock In {\em {SPIRE}}, volume 7608 of {\em Lecture Notes in Computer
  Science}, pages 379--384. Springer, 2012.
\newblock \href {https://doi.org/10.1007/978-3-642-34109-0\_40}
  {\path{doi:10.1007/978-3-642-34109-0\_40}}.

\bibitem{PaceT13}
Matthew~Felice Pace and Alexander Tiskin.
\newblock Parallel suffix array construction by accelerated sampling.
\newblock In {\em {PSC}}, pages 142--156. Department of Theoretical Computer
  Science, Faculty of Information Technology, Czech Technical University in
  Prague, 2013.

\bibitem{MaximilianParis2026}
Maximilian Paris.
\newblock Parallele konstruktion der burrows-wheeler-transformation auf der
  gpu.
\newblock Master's thesis, University of Münster, Münster, 2026.

\bibitem{PatelZMDO12}
Ritesh~A. Patel, Yao Zhang, Jason Mak, Andrew Davidson, and John~D. Owens.
\newblock Parallel lossless data compression on the {GPU}.
\newblock In {\em Innovative Parallel Computing (InPar)}, pages 1--9, 2012.
\newblock \href {https://doi.org/10.1109/InPar.2012.6339599}
  {\path{doi:10.1109/InPar.2012.6339599}}.

\bibitem{PingaliTB17}
Kartheek~Diwakar Pingali, K.~Chaitanya~Pavan Tanay, and Pallav~Kumar Baruah.
\newblock {GPU} accelerated suffix array construction for large genome
  sequences.
\newblock In {\em International Conference on Innovations in Information,
  Embedded and Communication Systems (ICIIECS)}, pages 1--6, 2017.
\newblock \href {https://doi.org/10.1109/ICIIECS.2017.8276035}
  {\path{doi:10.1109/ICIIECS.2017.8276035}}.

\bibitem{PiochowiakKD25}
Max Piochowiak, Florian Kurpicz, and Carsten Dachsbacher.
\newblock Random access segmentation volume compression for interactive volume
  rendering.
\newblock {\em Comput. Graph. Forum}, 44(3), 2025.
\newblock \href {https://doi.org/10.1111/cgf.70116}
  {\path{doi:10.1111/cgf.70116}}.

\bibitem{PolicritiP18}
Alberto Policriti and Nicola Prezza.
\newblock {LZ77} computation based on the run-length encoded {BWT.}
\newblock {\em Algorithmica}, 80(7):1986--2011, 2018.
\newblock \href {https://doi.org/10.1007/s00453-017-0327-z}
  {\path{doi:10.1007/s00453-017-0327-z}}.

\bibitem{PuglisiST07}
Simon~J. Puglisi, William~F. Smyth, and Andrew Turpin.
\newblock A taxonomy of suffix array construction algorithms.
\newblock {\em ACM Comput. Surv.}, 39(2):4, 2007.
\newblock \href {https://doi.org/10.1145/1242471.1242472}
  {\path{doi:10.1145/1242471.1242472}}.

\bibitem{RossiOLGB22}
Massimiliano Rossi, Marco Oliva, Ben Langmead, Travis Gagie, and Christina
  Boucher.
\newblock {MONI}: A pangenomic index for finding maximal exact matches.
\newblock {\em J. Comput. Biol.}, 29(2):169--187, 2022.
\newblock \href {https://doi.org/10.1089/cmb.2021.0290}
  {\path{doi:10.1089/cmb.2021.0290}}.

\bibitem{ShunBFGKST12}
Julian Shun, Guy~E. Blelloch, Jeremy~T. Fineman, Phillip~B. Gibbons, Aapo
  Kyrola, Harsha~Vardhan Simhadri, and Kanat Tangwongsan.
\newblock Brief announcement: the problem based benchmark suite.
\newblock In {\em {SPAA}}, pages 68--70. ACM, 2012.
\newblock \href {https://doi.org/10.1145/2312005.2312018}
  {\path{doi:10.1145/2312005.2312018}}.

\bibitem{Siren16}
Jouni Sir{\'{e}}n.
\newblock Burrows-wheeler transform for terabases.
\newblock In {\em {DCC}}, pages 211--220. IEEE, 2016.
\newblock \href {https://doi.org/10.1109/DCC.2016.17}
  {\path{doi:10.1109/DCC.2016.17}}.

\bibitem{Sun14}
Weidong Sun.
\newblock Using {GPU} to accelerate suffix array construction.
\newblock In {\em {BMEI}}, pages 677--682. IEEE, 2014.
\newblock \href {https://doi.org/10.1109/BMEI.2014.7002859}
  {\path{doi:10.1109/BMEI.2014.7002859}}.

\bibitem{SunM09}
Weidong Sun and Zongmin Ma.
\newblock Parallel lexicographic names construction with {CUDA.}
\newblock In {\em {ICPADS}}, pages 913--918. IEEE Computer Society, 2009.
\newblock \href {https://doi.org/10.1109/ICPADS.2009.31}
  {\path{doi:10.1109/ICPADS.2009.31}}.

\bibitem{Tischler17}
German Tischler.
\newblock Faster average case low memory semi-external construction of the
  burrows-wheeler transform.
\newblock {\em Math. Comput. Sci.}, 11(2):159--176, 2017.
\newblock \href {https://doi.org/10.1007/s11786-017-0296-2}
  {\path{doi:10.1007/s11786-017-0296-2}}.

\bibitem{WangBO16}
Leyuan Wang, Sean Baxter, and John~D. Owens.
\newblock Fast parallel skew and prefix-doubling suffix array construction on
  the {GPU.}
\newblock {\em Concurr. Comput. Pract. Exp.}, 28(12):3466--3484, 2016.
\newblock \href {https://doi.org/10.1002/cpe.3867}
  {\path{doi:10.1002/cpe.3867}}.

\bibitem{XieNLX20}
Jing~Yi Xie, Ge~Nong, Bin Lao, and Wentao Xu.
\newblock Scalable suffix sorting on a multicore machine.
\newblock {\em IEEE Trans. Computers}, 69(9):1364--1375, 2020.
\newblock \href {https://doi.org/10.1109/TC.2020.2972546}
  {\path{doi:10.1109/TC.2020.2972546}}.

\bibitem{ZakeriBAGL24}
Mohsen Zakeri, Nathaniel~K. Brown, Omar~Y. Ahmed, Travis Gagie, and Ben
  Langmead.
\newblock Movi: A fast and cache-efficient full-text pangenome index.
\newblock {\em iScience}, 27(12):111464, 2024.
\newblock \href {https://doi.org/10.1016/j.isci.2024.111464}
  {\path{doi:10.1016/j.isci.2024.111464}}.

\bibitem{ZakeriBGL26}
Mohsen Zakeri, Nathaniel~K. Brown, Travis Gagie, and Ben Langmead.
\newblock Movi 2: fast and space-efficient queries on pangenomes.
\newblock {\em Bioinform.}, 42(7), 2026.
\newblock \href {https://doi.org/10.1093/bioinformatics/btag362}
  {\path{doi:10.1093/bioinformatics/btag362}}.

\bibitem{ZhaoYX12}
Zhiheng Zhao, Jianping Yin, and Wei Xiong.
\newblock {GPU}-accelerated burrow-wheeler transform for genomic data.
\newblock In {\em {BMEI}}, pages 889--892. IEEE, 2012.
\newblock \href {https://doi.org/10.1109/BMEI.2012.6513199}
  {\path{doi:10.1109/BMEI.2012.6513199}}.

\end{thebibliography}

\appendix

\section{More Related Work}\label{sec:more_related_work}
This appendix covers BWT construction and suffix sorting in other models of computation.

\subparagraph*{BWT construction in other models.}
In the sequential setting, the BWT can be constructed directly from the text in compact space~\cite{HonLSSY07,HonSS09}, by blockwise suffix sorting in small space~\cite{Karkkainen07}, by induced sorting in linear time~\cite{OkanoharaS09}, in linear time and $O(n\log\sigma)$ bits~\cite{Belazzougui14,MunroNN17}, in sublinear time for small alphabets~\cite{KempaK19}, and even in place~\cite{CrochemoreGKL15}.
Lyndon-based approaches sort the suffixes of a text or a string collection via the Lyndon factorization~\cite{MantaciRRS14,BonomoMRRS14}, construct the bijective and the extended BWT in linear time~\cite{BannaiKKP24} or in place~\cite{KopplHHS20}, and construct the BWT of repetitive texts from a grammar based on the standard factorization of Lyndon words~\cite{Olbrich25}.
For repetitive texts, the run-length BWT can be maintained online in space proportional to $r$~\cite{PolicritiP18,OhnoSTIS18} and constructed in optimal time within the same space~\cite{NishimotoKT22}, in optimal time for highly repetitive texts~\cite{Kempa19}, and from the LZ77 parse~\cite{KempaK20}.
In practice, repetitive texts are handled by prefix-free parsing~\cite{BoucherGKLMM19}, which has been extended to the suffix array and the r-index~\cite{KuhnleMBGLM20} and applied recursively to the parse itself~\cite{OlivaGB23}.
Other approaches use grammar compression~\cite{Diaz-DominguezN23}, matching statistics~\cite{Masillo23}, and the incremental construction of ropebwt3~\cite{Li24}, which indexes hundreds of human genomes.
In external and semi-external memory, the BWT is constructed with sequential scans only~\cite{FerraginaGM12}, with little main memory~\cite{BellerZGO13,Tischler17}, by BCR~\cite{BauerCR13,CoxGRS16} and other algorithms for string collections~\cite{BonizzoniVPPR21,AdlerBHS25}, and by merging the BWTs of subcollections~\cite{HoltM14,Siren16,EgidiLMT19,EgidiM20}.
Shared-memory parallel constructions exist for single genomes~\cite{HayashiT13,LiuHS16}, for collections of genomes~\cite{LiuZW16}, and in compact space~\cite{Fuentes-SepulvedaNN19}. PRAM algorithms construct the BWT with polylogarithmic depth~\cite{EdwardsV14,Fuentes-SepulvedaNN20}. In distributed memory, the BWT has been computed with MapReduce~\cite{MenonBS11} and Apache Spark~\cite{GalluzzoGRR26}.
All PFP-based constructions we are aware of run on CPUs and rely on hash tables and inverted lists.

\subparagraph*{Suffix sorting in other models.}
Puglisi et al.~\cite{PuglisiST07} survey sequential SACAs. SACABench~\cite{SACABench19} compares many of them.
The classical approaches are prefix doubling~\cite{ManberM93,LarssonS07}, the recursive linear-time algorithms~\cite{KarkkainenSB06,KimSPP05,KoA05}, and induced sorting~\cite{NongZC11,Nong13}. More recent algorithms are non-recursive and based on Lyndon words~\cite{Baier16,BertramEF21,OlbrichOB24} or work in place in linear time~\cite{LiLH22}.
Practical implementations include deep-shallow sorting~\cite{ManziniF04}, \code{divsufsort}~\cite{divsufsort,FischerK17}, and \libsais~\cite{libsais}, which is based on induced sorting and among the fastest.
In external memory, suffix arrays are constructed by doubling~\cite{CrauserF02,DementievKMS08}, by the skew algorithm~\cite{DementievKMS08}, by induced sorting~\cite{BingmannFO16,NongCHW15,KarkkainenKPZ17}, and in parallel~\cite{KarkkainenKP15}.
Shared-memory parallel SACAs parallelize the skew algorithm~\cite{ShunBFGKST12}, prefix doubling and induced copying~\cite{LabeitSB17}, induced sorting~\cite{LaoNCX18,XieNLX20}, and samplesort~\cite{KhanRMDP24}.
In distributed memory, SACAs parallelize doubling~\cite{KitajimaN99,FlickA15} and the skew algorithm~\cite{KullaS07,BingmannGK18}. Recent lightweight constructions scale to many compute nodes~\cite{FischerK19,HaagKSS25}.
The skew algorithm also has PRAM and BSP variants~\cite{KarkkainenSB06}. Pace and Tiskin~\cite{PaceT13} give a deterministic BSP algorithm with few synchronization steps.

\section{Additional Experimets}
\subparagraph*{The memory ladder.}
\Cref{fig:phases} shows the time per phase (the eight phases of \cref{fig:gpu-flow} plus setup and output, which include the transfers between host and GPU) of all configurations at the largest input size that all of them handle, and \cref{fig:spacetime-all} shows the same runs as space-time points.
On the non-repetitive inputs, each level lowers the peak. On \texttt{english} at $381.5$\,MiB, the peak drops from $47.7$ bytes per character at level~0 to $41.0$ at level~2, $32.0$ at level~3, $20.9$ at level~4, and $13.4$ and $11.9$ at levels~5 and~6.
The levels cost time: level~4 is the slowest on every non-repetitive input (on \texttt{pitches}, $19.1$\,s instead of $3.2$\,s at level~0), and levels~5 and~6 recover only part of the loss ($9.1$\,s).
The MSD core removes this cost.
On \texttt{english}, \code{cub6} has the same peak as \code{opt6} ($11.9$ bytes per character) and needs $1.6$ instead of $2.6$\,s.
On \texttt{pitches}, it needs $1.2$ instead of $9.1$\,s at $17.0$ instead of $14.2$ bytes per character.
On six of the nine repetitive inputs, the peak does not change across the levels, because $s$ is small (at most $0.10n$) and the peak lies in the dictionary phases (parsing, the two dictionary sorts, and the parse ranks).
There, \code{cub6} needs more memory than \code{opt6} because of the scratch of its unchunked dictionary sorts.
\code{space} removes this scratch and lies on the Pareto front on $13$ of the $15$ inputs. The levels below~5 are dominated on $14$ of them.
On the whole \texttt{english}, \code{space} lowers the peaks of the LR dictionary sort ($6.7 \to 4.6$ bytes per character) and of the suffix generation ($7.8 \to 5.8$). The suffix sort then determines the overall peak, which drops only from $10.0$ to $9.9$.
Where a dictionary sort determines the peak, \code{space} lowers it from $15.3$ to $11.1$ bytes per character on the whole \texttt{proteins} and by $24$ to $27\%$ on eight of the nine repetitive inputs.

\begin{figure}[tbp]
\centering
\includegraphics[width=0.81\textwidth]{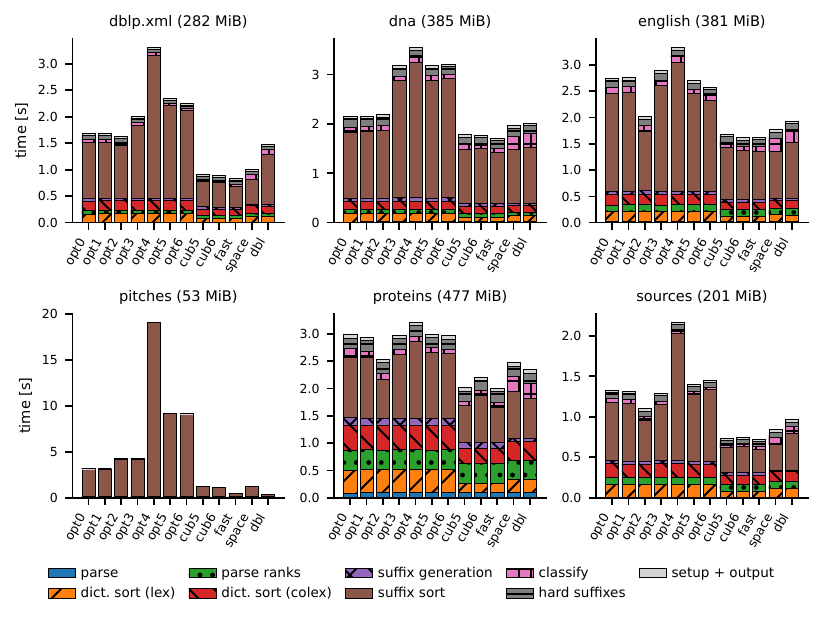}\\[1ex]
\includegraphics[width=0.81\textwidth]{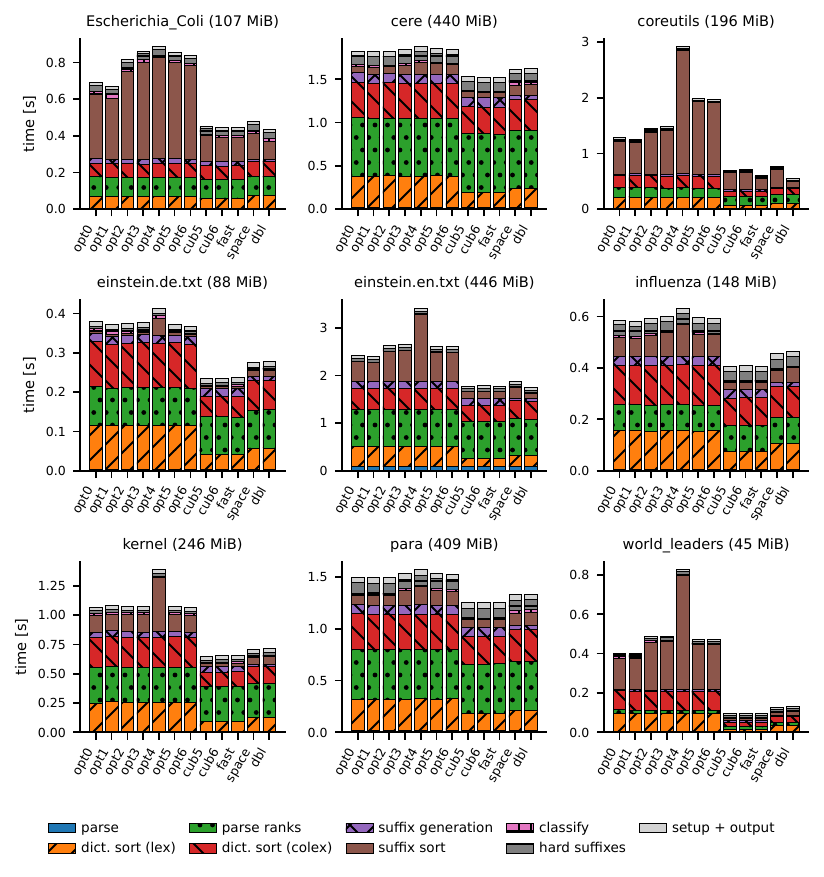}
\caption{Time per phase at the largest input size that all configurations handle, for the non-repetitive (top) and the repetitive texts (bottom).}\label{fig:phases}
\end{figure}

\section{Comparison with Big-BWT}\label{app:bigbwt}

\Cref{tab:mapping} maps each step of Big-BWT (\cref{sec:pfp-impl}) to the corresponding step of our algorithm (\cref{sec:algo}).

\begin{table}[tbp]
\caption{The steps of the original PFP-BWT (Big-BWT, \cref{sec:pfp-impl}) and
of our GPU algorithm.}\label{tab:mapping}
\centering\footnotesize
\setlength{\tabcolsep}{4pt}
\begin{tabular}{@{}>{\raggedright\arraybackslash}p{2.3cm}>{\raggedright\arraybackslash}p{5.2cm}>{\raggedright\arraybackslash}p{5.4cm}@{}}
\toprule
step & original (Big-BWT) & ours \\
\midrule
text and \newline conventions
  & $\#\,T\,\$^w$ (in practice NUL padding), BWT of $T\$$
  & $T\,\$^w$, the first window $T[0..w-1]$ acts as a trigger string instead of
    $\#$, BWT in the convention of \libcubwt \\
\addlinespace
triggers
  & rolling Karp--Rabin hash, sequential scan
  & the same fingerprint, evaluated directly per window, no dependency between windows \\
\addlinespace
parsing
  & one pass emitting phrases
  & count triggers per block, scan, write phrase ends into $\Docc$ \\
\addlinespace
dictionary
  & hash table of phrase hashes with frequencies, sorted dictionary on disk
  & never materialized, groups of equal phrases in the LR order of $\Docc$ \\
\addlinespace
parse
  & phrase ranks on disk
  & group ranks of the LR order, written back in text order \\
\addlinespace
order of the parse
  & $\SA(\Parse)$ by SACA-K, inverted lists from $\SA(\Parse)$, $W$ permuted
  & $\SA(\Parse)$ by prefix doubling, and per occurrence the rank of the
    following parse suffix \\
\addlinespace
phrase suffixes
  & SA and LCP of the concatenated dictionary by gSACA-K
  & generated once each by adjacency in the RL order, then sorted as strings
    (\cref{sec:memory,sec:refinements}) \\
\addlinespace
occurrence sets
  & phrases sharing a suffix found by the LCP scan, their occurrences found via
    the inverted lists
  & a contiguous interval of the RL order, found by a next-smaller-value query
    on the longest common suffixes \\
\addlinespace
easy/hard test
  & compare the preceding characters of all phrases sharing the suffix
  & compare the preceding characters of the two interval endpoints \\
\addlinespace
emission
  & easy suffixes as a run of the summed phrase frequencies, hard suffixes by a heap merge
    of inverted lists
  & easy suffixes as a run of the interval size, hard suffixes by sorting the
    interval by parse rank \\
\addlinespace
whole phrases
  & preceding character from $W$
  & preceding character read from $T$ \\
\addlinespace
execution
  & sequential, external memory, space about $|\Parse|$ plus the total length of $\Dset$
  & GPU-resident, data-parallel, space governed by $d$ and $s$ \\
\bottomrule
\end{tabular}
\end{table}

\end{document}